\documentclass[reqno]{amsart}

\usepackage{amsmath,amsthm,amsfonts,amssymb,srcltx}
\usepackage{graphicx} 
\usepackage{epstopdf}
\usepackage{enumerate}
\usepackage{cite}

\usepackage{color}

\newcommand{\bea}{\begin{eqnarray}}
\newcommand{\eea}{\end{eqnarray}}
\newcommand{\be}{\begin{equation}}
\newcommand{\ee}{\end{equation}}

\newtheorem{theorem}{Theorem}[section]
\newtheorem{proposition}[theorem]{Proposition}

\newtheorem{corollary}[theorem]{Corollary}
\theoremstyle{definition}
\newtheorem{definition}[theorem]{Definition}

\newtheorem{example}[theorem]{Example}

\renewenvironment{proof}{{\noindent\bf Proof.}}{\hfill $\Box$\par\vskip3mm}

\makeatletter
 
 \@addtoreset{equation}{section}
\makeatother

\newcommand{\e}{{\rm e}}

\def\IR{\mathbb {R}}

\begin{document}

\date{}

\author[J.E. Andersen]{J{\o}rgen Ellegaard Andersen}
\address{QGM\\
Department of Mathematics\\
Aarhus University\\
DK-8000 Aarhus C\\
Denmark}
\email{jea.qgm{\char'100}gmail.com}

\author[H. Fuji]{Hiroyuki Fuji}
\address{Faculty of Education\\ 
Kagawa University\\
Takamatsu 760-8522\\
Japan; 
QGM\\
Aarhus University\\
DK-8000 Aarhus C\\
Denmark}
\email{fuji{\char'100}ed.kagawa-u.ac.jp}

\author[M. Manabe]{Masahide Manabe}
\address{Faculty of Physics\\
University of Warsaw\\
ul. Pasteura 5, 02-093 Warsaw\\
Poland}
\email{masahidemanabe{\char'100}gmail.com}

\author[R. C. Penner]{Robert C. Penner}
\address{Institut des Hautes {\'E}tudes Scientifiques, 35 route de Chartres, 91440 Burs-sur-Yvette, France;
Division of Physics, Mathematics and Astronomy, California Institute of Technology, Pasadena, CA 91125, USA}
\email{rpenner{\char'100}caltech.edu,\hspace{0.3cm}rpenner@ihes.fr}

\author[P. Su{\l}kowski]{Piotr Su{\l}kowski}
\address{Faculty of Physics, University of Warsaw, ul. Pasteura 5, 02-093 Warsaw, Poland; Walter Burke Institute for Theoretical Physics, California Institute of Technology, Pasadena, CA 91125, USA}
\email{psulkows{\char'100}fuw.edu.pl}

\title [Enumeration of chord diagrams]{Enumeration of chord diagrams via topological recursion and quantum curve techniques}

\thanks{Acknowledgments: 
JEA and RCP are supported by the Centre 
for Quantum Geometry of Moduli Spaces which is funded by the Danish National Research Foundation.
The research of HF is supported by the
Grant-in-Aid for Research Activity Start-up [\# 15H06453], Grant-in-Aid
for Scientific Research(C)  [\# 26400079], and Grant-in-Aid for Scientific
Research(B)  [\# 16H03927]  from the Japan Ministry of Education, Culture,
Sports, Science and Technology.
The work of MM and PS is supported by the ERC Starting Grant no. 335739 ``Quantum fields and knot homologies'' funded by the European Research Council under the European Union's Seventh Framework Programme.
PS also acknowledges the support of the Foundation for Polish Science, and RCP acknowledges the kind support of
Institut Henri Poincar\'e where parts of this manuscript were written.}
\begin{abstract}
In this paper we consider the enumeration of orientable and non-orientable chord diagrams. We show that this enumeration is encoded in appropriate expectation values of the  $\beta$-deformed Gaussian and RNA matrix models. We evaluate these expectation values by means of the $\beta$-deformed topological recursion, and -- independently -- using properties of quantum curves. We show that both these methods provide efficient and systematic algorithms for counting of chord diagrams with a given genus, number of backbones and number of chords.
\end{abstract}

\maketitle
\tableofcontents

\section{Introduction} \label{sec:chord_review}

A \textit{chord diagram} 
is a graph which can be realised in the plane as follows. It is 
comprised of a collection of $b$ line segments (called \textit{backbones}) on the real axis
 with $k$ semi-circles (called \textit{chords}) in the upper-half plane attached to the line segments.
All chords are attached at different points on the backbones. 
A chord diagram comes from its realisation in the plane 
with a natural \textit{fatgraph} structure, namely, 
half edges incident to each trivalent vertex are endowed with a cyclic order induced from the orientation of the plane.
For a chord diagram $c$ the fatgraph structure allows us, in the usual way, to define a surface $\Sigma_c$, which is simply just a small tubular neighbourhood of the realisation of the chord diagram in the plane, see left and middle diagrams in Figure \ref{chord_fat}.
Let $n$ be the number of boundary cycles,
and $g$ be the genus of the skinny surface $\Sigma_c$. Then
the Euler characteristic $\chi$ of $\Sigma_c$ is given by
\begin{align}
\chi=2-2g=b-k+n.
\nonumber
\end{align}

\begin{figure}[h]
\begin{center}
   \includegraphics[width=120mm,clip]{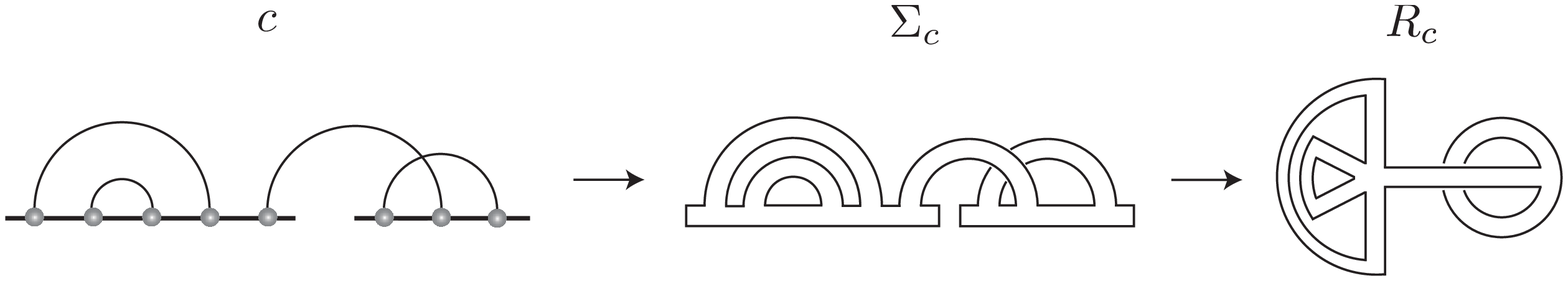}
\end{center}
\caption{\label{chord_fat} A chord diagram $c$, its skinny surface $\Sigma_c$, and the associated ribbon surface $R_c$.}
\end{figure}

To present a chord diagram $c$ more simply and efficiently, we collapse each
fattened backbone in $\Sigma_c$ into a polyvalent fattened vertex, see the right side of Figure \ref{chord_fat}.
We call the resulting surface $R_c$ and we refer to it as the \textit{Ribbon surface} associated to $c$.
In order for $c$ to be uniquely determined by the ribbon surface $R_c$, at each vertex we attach a tail at the place corresponding to the beginning of the backbone, see Figure \ref{tail_vertex}.

\begin{figure}[h]
\begin{center}
   \includegraphics[width=80mm,clip]{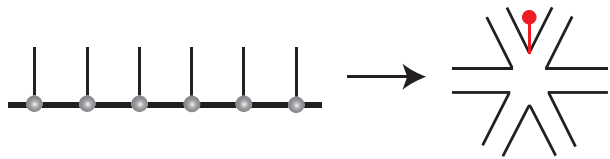}
\end{center}
\caption{\label{tail_vertex} Polyvalent fatten vertex with a tail}
\end{figure}

In this paper we consider the enumeration of chord diagrams with the topological filtration induced by the genus and the number of 
backbones, employing matrix model techniques.
Chord diagrams are widely used objects in pure and applied mathematics, see e.g. 
\cite{Ba,Kontsevich2,AMR1,AMR2,ABMP,C-SM}.
Furthermore, they are now used also in the biological context for characterisation of secondary and pseudoknot structures of RNA molecules
\cite{PW,P4,OZ,VOZ,VO,APRWat,Reidys_book}.\footnote{
The combinatorial aspects of interacting RNA molecules with the associated
genus filtration are also discussed in \cite{Bon08,POZ,PTOZ,VROZ,APRWan,APRWat,APRW,AAPZ,AFMPS,AFPR},
and folding algorithms are studied in \cite{BO,AHPR,RHAPSN}.
The matrix model approach to the enumeration of possible RNA structures is also studied in
\cite{BhDe1,BhDe2,GBD1,GBD2,GaDe1,GaDe2,GaDe3,GaDe4,GaDe5}.
}
In particular, motivated by the study of RNA pseudoknot structures, 
a matrix model for the enumeration of chord diagrams -- which we refer to as the RNA-matrix model -- was constructed in \cite{ACPRS,ACPRS2}. 
In this paper we study the $\beta$-deformed version of this model and present how it encodes orientable and non-orientable chord diagrams.

\subsection{The RNA matrix model for orientable chord diagrams} \label{subsec:ori_chord}

Let $c_{g,b}(k)$ denote the number of connected chord
diagrams with genus $g$, $b$ backbones, and $k$ chords.\footnote{Harer and Zagier found a remarkable formula for $c_{g,1}(k)$, referred to as the \textit{Harer-Zagier formula}, in their computation of the virtual Euler characteristic of Riemann moduli space
for punctured surfaces \cite{HarZagier}.
}
We consider the following generating function  
\begin{align}
C_{g,b}(w)=\sum_{k=0}^{\infty}c_{g,b}(k)w^{k}.
\label{gen_fn_chord}
\end{align}
In \cite{APRWat}, 
the number $c_{g,1}(k)$ of chord diagrams with 1 backbone was
studied. In particular for the class of planar graphs which have genus $g=0$,
the number $c_{0,1}(k)$ is shown to be equal to the Catalan number
\begin{align}
C_k=\frac{(2k)!}{(k+1)!k!}.
\nonumber
\end{align}
We present explicitly the tailed ribbon surfaces with 1 backbone for
$k=1,2,3,4$ in Figure \ref{tail_planar}.
\begin{figure}[h]
\begin{center}
   \includegraphics[width=120mm,clip]{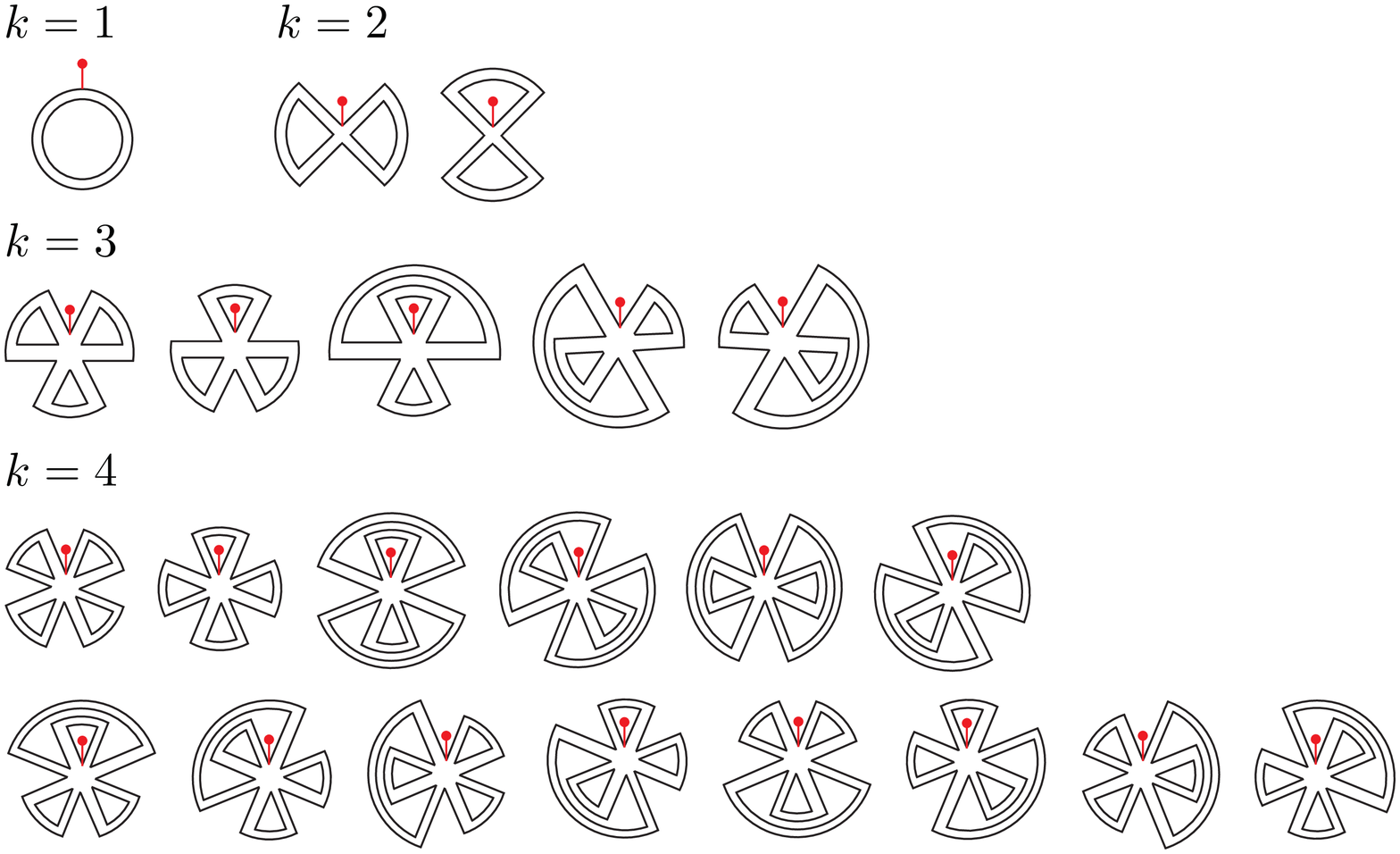}
\end{center}
\caption{\label{tail_planar} The planar ribbon surfaces with tails for
 $k=1,2,3,4$, whose counts agree with the Catalan numbers $1$,
 $2$, $5$, $14$.}
\end{figure}

In \cite{ACPRS} the following theorem was established.
\begin{theorem}[RNA matrix model for orientable chord diagrams \cite{ACPRS}]\label{thm:rna_matrix}
Let $\mathcal{H}_N$ be the space of rank $N$ Hermitian matrices. We consider the matrix integral with the potential
\begin{equation}
V_{\mathrm{RNA}}(x)=\frac{x^2}{2}-\frac{stx}{1-tx},
\label{potential_RNA}
\end{equation}
defined by
\begin{equation}
Z_N(s,t)=\frac{1}{\mathrm{Vol}_N}\int_{\mathcal{H}_N}dM\;\e^{-\frac{1}{\hbar}{\mathrm{Tr}}V_{\mathrm{RNA}}(M)},
\label{matrix_integral}
\end{equation}
where 
\begin{align}
\mathrm{Vol}_N=\int_{{\mathcal H}_N} dM\;\mathrm{e}^{-N\mathrm{Tr}\frac{M^2}{2}}
=N^{N(N+1)/2}\mathrm{Vol}({\mathcal H}_N).
\end{align}
Under the 't Hooft limit
\begin{equation}
\hbar\to 0,\ \ \ \ N\to \infty,\ \ \ \ \mu=\hbar N,
\end{equation}
with the 't Hooft parameter $\mu$ kept finite, the logarithm of the above matrix integral (called the \textit{free energy}) $F_N(s,t)=\log Z_N(s,t)$ has an asymptotic expansion
\begin{equation}
F_N(s,t)=\sum_{g=0}^{\infty}\hbar^{2g-2}F_g(s,t).
\label{Fg_asympt_g}
\end{equation}
Moreover, this free energy encodes generating functions (\ref{gen_fn_chord}) 
for the numbers $c_{g,b}(k)$ of chord diagrams 
\begin{equation}
\mu^{2g-2}F_g(s,t)=\sum_{b=1}^{\infty}\frac{s^b}{\mu^{b}b!}C_{g,b}(\mu t^2)-\frac{s}{\mu}\delta_{g,0}.
\label{F_C_orientable}
\end{equation}
\end{theorem}

By (\ref{F_C_orientable}) we see that $F_g(s,t)$ enumerates orientable chord diagrams with genus $g$, and $s$ and $t^2$ are generating parameters respectively for the number of backbones and chords.

\subsection{RNA matrix model for non-oriented chord diagrams} \label{subsec:non_ori}

As a natural generalization of the above enumerative problem, we consider the non-oriented\footnote{Non-oriented is a shorthand for the union of orientable and non-orientable.} analogue of the enumeration of chord diagrams. A non-oriented chord diagram $c$ is a chord diagram, where a binary quantity, twisted or untwisted, is assigned to each chord. There is a natural non-oriented fatgraph structure associated with a non-oriented chord diagram and the corresponding non-oriented surface $\Sigma_c$, where the binary quantity assigned to each chord determines if the band along the chord for the associated ribbon surface is twisted or not, as depicted in Figure \ref{twist_chord}.

The non-oriented analogue of the ribbon surface is again constructed by collapsing the fattened backbones into polyvalent fattened vertices and, as before, adding a tail for each fattened vertex. Some examples of  non-oriented ribbon surfaces are depicted in Figure \ref{non_orientable_4} and in Appendix \ref{app:non_ori_fatgraph}.

\begin{figure}[h]
\begin{center}
   \includegraphics[width=60mm,clip]{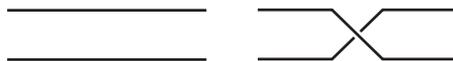}
\end{center}
\caption{\label{twist_chord} A band and a twisted band.}
\end{figure}

\begin{figure}[h]
\begin{center}
   \includegraphics[width=100mm,clip]{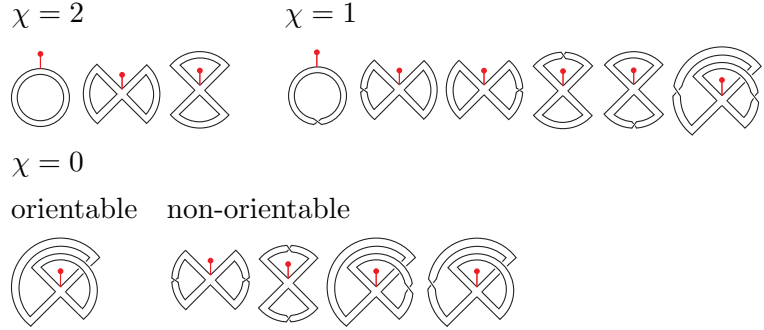}
\end{center}
\caption{\label{non_orientable_4} Non-oriented ribbon surfaces with 1 backbone 
for $k=1,2$.}
\end{figure}

Instead of the genus, for a surface $\Sigma_c$ the \textit{cross-cap number} (or the \textit{non-oriented genus}) $h$ is well-defined, and 
the Euler characteristic is given by
\begin{align}
\chi=2-h=b-k+n.
\nonumber
\end{align}

Let $c^{\mathsf r}_{h,b}(k)$ denote the number 
of non-oriented chord diagrams with the cross-cap number $h$, $b$ backbones, and $k$ chords.
In analogy to the oriented case, we introduce the generating
function $C^{\mathsf r}_{h,b}(w)$ 
\begin{align}
C^{\mathsf r}_{h,b}(w)=\sum_{k=0}^{\infty}c^{\mathsf r}_{h,b}(k)w^k.
\label{gen_fn_non_orientable}
\end{align}

As proven in \cite{AAPZ}, the non-oriented analogue of Theorem \ref{thm:rna_matrix} is straightforwardly obtained by replacing the integration over Hermitian matrices in (\ref{matrix_integral}) with the integration over real symmetric matrices. 

\begin{theorem}[RNA matrix model for non-oriented chord diagrams]\label{thm:rna_matrix_non}
Let $\mathcal{H}_N({\IR})$ be the space of rank $N$ real symmetric matrices.
We consider the real symmetric matrix integral with the potential (\ref{potential_RNA}) defined by
\begin{equation}
Z_N^{\mathsf{r}}(s,t)=\frac{1}{\mathrm{Vol}_N({\IR})}\int_{\mathcal{H}_N({\IR})}dM\;\e^{-\frac{1}{2\hbar}{\mathrm{Tr}}V_{\mathrm{RNA}}(M)},
\label{real_matrix_integral}
\end{equation}
where
\begin{align}
\mathrm{Vol}_N({\IR})=
\int_{\mathcal{H}_N({\IR})}dM\;\mathrm{e}^{-N\mathrm{Tr}\frac{M^2}{4}}
=N^{N(N+1)/2}\mathrm{Vol}(\mathcal{H}_N({\IR})).
\end{align}
Under the 't Hooft limit
\begin{equation}
\hbar\to 0,\ \ \ \ N\to \infty,\ \ \ \ \mu=\hbar N,
\end{equation}
with the 't Hooft parameter $\mu$ kept finite and fixed, the free energy $F_N^{\mathsf{r}}(s,t)=\log Z_N^{\mathsf{r}}(s,t)$ 
has an asymptotic expansion
\begin{equation}
F_N^{\mathsf{r}}(s,t)=\frac12\sum_{h=0}^{\infty}\hbar^{h-2}F_h^{\mathsf{r}}(s,t).
\label{Fg_asympt_h}
\end{equation}
This free energy encodes generating functions (\ref{gen_fn_non_orientable}) 
for the numbers $c_{h,b}^{\mathsf{r}}(k)$ of non-oriented chord diagrams
\begin{equation}
\mu^{h-2}F_h^{\mathsf{r}}(s,t)=\sum_{b=1}^{\infty}\frac{s^b}{\mu^{b}b!}C_{h,b}^{\mathsf{r}}(\mu t^2)
-\frac{s}{\mu}\delta_{g,0}.
\label{F_C_non_orientable}
\end{equation}
\end{theorem}

Note that matrix model techniques for the enumeration of non-oriented chord diagrams are also considered in
\cite{GouHarJack,VWZ,Si,JNPI,MW,GM}.

\subsection{$\beta$-deformed RNA matrix model as a unified model} \label{subsec:beta_def}

In the context of matrix models, it is known that their $\beta$-deformation implements the enumeration of non-oriented chord diagrams \cite{Mehta}, as we shall now recall\footnote{The $\beta$-deformed Dyson's model is solved in various ways. See e.g. \cite{Katori1,Katori2,Katori3}.}.

\begin{definition}[$\beta$-deformed RNA matrix model]\label{def:beta_RNA}
The $\beta$-deformed eigenvalue integral with the potential (\ref{potential_RNA}) is defined by
\begin{equation}
Z_N^{\beta}(s,t)=\frac{1}{\mathrm{Vol}_N^{\beta}}\int_{{\IR}^N} \prod_{a=1}^Ndz_a \Delta(z)^{2\beta}\e^{-\frac{\sqrt{\beta}}{\hbar}\sum_{a=1}^NV_{\mathrm{RNA}}(z_a)},
\label{eigenvalue_integral_beta}
\end{equation}
where $\Delta(z)$ denotes the \textit{Vandermonde determinant}
\begin{equation}
\Delta(z)=\prod_{a<b}(z_a-z_b),
\nonumber
\end{equation}
and 
\begin{align}
\mathrm{Vol}_N^{\beta}=
\int_{{\IR}^N} \prod_{a=1}^Ndz_a \Delta(z)^{2\beta}\e^{-\frac{\sqrt{\beta}}{2\hbar}\sum_{a=1}^Nz_a^2}.
\end{align}
\end{definition}

In the cases of $\beta=1$ and $\beta=1/2$, the $\beta$-deformed eigenvalue integral (\ref{eigenvalue_integral_beta}) reduces to the eigenvalue representation of the Hermitian matrix integral (\ref{matrix_integral}) and the real symmetric matrix integral (\ref{real_matrix_integral}) upon the redefinition $\hbar \to \sqrt{2}\hbar$, respectively. Here $z_a\in \mathbb{R}$ ($a=1,\cdots,N$)  correspond to the eigenvalues of the matrix $M$ in each matrix integral.
The other special case is $\beta=2$, for which the eigenvalue integral (\ref{eigenvalue_integral_beta}) represents the quaternionic matrix integral. Under the 't Hooft limit
\begin{equation}
\hbar\to 0,\ \ \ \ N\to\infty,\ \ \ \ \mu=\beta^{1/2}\hbar N,
\label{limit_beta}
\end{equation}
with the fixed 't Hooft parameter $\mu$, 
the free energy $F_N^{\beta}(s,t)=\log Z_N^{\beta}(s,t)$ has the asymptotic expansion
\begin{equation}
F_N^{\beta}(s,t)=\sum_{g,\ell=0}^{\infty}\hbar^{2g-2+\ell}\gamma^{\ell}F_{g,\ell}(s,t),
\label{free_energy_beta}
\end{equation}
where
\begin{equation}
\gamma=\beta^{1/2}-\beta^{-1/2}.
\nonumber
\end{equation}
The free energies (\ref{Fg_asympt_g}) for $\beta=1$ and (\ref{Fg_asympt_h}) for $\beta=1/2$ satisfy
\begin{align}
\begin{split}
&
F_g(s,t)=F_{g,0}(s,t),
\\
&
F_h^{\mathsf{r}}(s,t)=\sum_{
\substack{
g,\ell=0\\
2g+\ell=h
}}^{\infty}2^g(-1)^{\ell}
F_{g,\ell}(s,t).
\end{split}
\end{align}

Combining (\ref{F_C_orientable}) in Theorem \ref{thm:rna_matrix} for $\beta=1$ and (\ref{F_C_non_orientable}) in Theorem \ref{thm:rna_matrix_non} for $\beta=1/2$, we obtain the following proposition.
\begin{proposition}\label{prop:rna_beta}
Let $\widetilde{C}_{g,\ell,b}(w)$ be defined by
\begin{equation}
(-\mu)^{2g-2+\ell}F_{g,\ell}(s,t)=\sum_{b=1}^{\infty}\frac{s^b}{\mu^bb!}\widetilde{C}_{g,\ell,b}(\mu t^2)
-\frac{s}{\mu}\delta_{g,0}\delta_{\ell,0}.
\label{non_ori_rel1}
\end{equation}
Then we have the following relations 
\begin{align}
\begin{split}
&
C_{g,b}(w)=\widetilde{C}_{g,0,b}(w),
\\
&
C_{h,b}^{\mathsf r}(w)=\sum_{
\substack{
g,\ell=0\\
2g+\ell=h
}}^{\infty}2^g\widetilde{C}_{g,\ell,b}(w).
\label{gen_chord_non}
\end{split}
\end{align}
\end{proposition}

In the following sections, using $\beta$-deformed topological recursion and the theory of quantum curves, we will develop the computational techniques of determining the functions $F_{g,\ell}$.

\subsection{$\beta$-deformed topological recursion and Gaussian resolvents}

In our first approach, we will employ an analytical method, referred to as the \textit{$\beta$-deformed topological recursion} \cite{Chekhov:2006rq,Brini:2010fc,Borot:2010tr,Marchal:2011iu,MS}.
$\beta$-deformed topological recursion is a powerful machinery for analyzing the asymptotic expansion of a matrix model free energy in the $N\to\infty$ limit.
Let $Z_N^{\beta}(\{r_n\})$ and $F_N^{\beta}(\{r_n\})=\log Z_N^{\beta}(\{r_n\})$ denote the partition function and the free energy of the $\beta$-deformed matrix model with a general potential
\begin{equation}
V(x)=\sum_{n=0}^Kr_nx^n.
\label{g_pot}
\end{equation}
By the asymptotic expansion of the free energy $F_N^{\beta}(\{r_n\})$
we mean the following expansion in the 't Hooft limit (\ref{limit_beta})
\begin{equation}
F_N^{\beta}(\{r_n\})=\sum_{g,\ell=0}^{\infty}\hbar^{2g-2+\ell}\gamma^{\ell}F_{g,\ell}(\{r_n\}).
\label{free_energy_beta_g}
\end{equation}

The $\beta$-deformed topological recursion in fact can be formulated more generally, as a tool that assigns a series of multi-linear differentials to a given algebraic curve. In the context of matrix models these multi-linear differentials are identified with various matrix model correlators (and the topological recursion represents Ward identities between these correlators), while the underlying algebraic curve is identified with a \textit{spectral curve} $\mathcal{C}$ (of the matrix model with $\beta=1$)
\begin{align}
\begin{split}
&
\mathcal{C}=\big\{\
(x,\omega)\in \mathbb{C}^2\ \big|\ H(x,\omega)=\omega^2-2V'(x)\omega-f(x)=0\
\big\},
\\
& 
f(x)=\lim_{\substack{
N\to \infty,\hbar\to 0\\
\beta\to 1}}
-4\beta^{1/2}\hbar\left\langle \sum_{a=1}^{N}
\frac{V'(x)-V'(z_a)}{x-z_a}
\right\rangle_N^{\beta},
\label{spectral}
\end{split}
\end{align}
where $\langle \mathcal{O}(\{z_a\})\rangle_N^{\beta}$ denotes
the $\beta$-deformed eigenvalue integral 
\begin{equation}
\langle \mathcal{O}(\{z_a\})\rangle_N^{\beta}
=\frac{1}{Z_N^{\beta}(\{r_n\})\mathrm{Vol}_N^{\beta}}
\int_{{\IR}^N} \prod_{a=1}^N dz_a \Delta(z)^{2\beta}\mathcal{O}(\{z_a\})\e^{-\frac{\sqrt{\beta}}{\hbar}\sum_{a=1}^N
V(z_a)}.
\label{expextation}
\end{equation}
In particular for the RNA matrix model presented in Definition \ref{def:beta_RNA},
the spectral curve takes the following form \cite{ACPRS}. 
\begin{example}
For the potential $V_{\mathrm{RNA}}(x)$ given in (\ref{potential_RNA}), the defining equation for 
${\mathcal C}$ in (\ref{spectral}) takes form
\begin{equation}
y_{\mathrm{RNA}}^2=M_{\mathrm{RNA}}(x)^2(x-a)(x-b),
\quad y_{\mathrm{RNA}}=V_{\mathrm{RNA}}'(x)-\omega,
\label{RNA_curve}
\end{equation}
where the branch points $a<b$ are solutions of the equations
\begin{align}
&
S=\frac{D^2-4\mu}{t(3D^2-4\mu)},
\nonumber
\\
&
4s^2t^4(3D^2-4\mu)^6=D^2(D^2-4\mu)^2\big(4D^2-t^2(3D^2-4\mu)^2\big)^3,
\nonumber
\end{align}
and $\sigma$ and $\delta$ are defined as
\begin{equation}
\sigma=St=\frac{a+b}{2}t,\ \ \ \
\delta=Dt=\frac{a-b}{2}t.
\nonumber
\end{equation}
The function $M_{\mathrm{RNA}}(x)$ is given by
\begin{equation}
M_{\mathrm{RNA}}(x)=\frac{(2tx-2+\sigma)^2+\eta}{8(tx-1)^2},\ \ \ \
\eta:=\frac{\sigma(4-4\delta^2-7\sigma+3\sigma^2)}{\sigma-1}.
\nonumber
\end{equation}
\end{example}

The multi-linear differentials computed by the $\beta$-deformed topological recursion are defined as follows \cite{Chekhov:2006rq}.
\begin{definition}\label{def:multi-linear}
The connected $h$-point symmetric multi-linear differential $W_h\in {\mathcal M}^{1}(\mathcal{C}^{\times h})^{s}$ is defined as
\begin{equation}
W_h(x_1,\ldots,x_h)=\beta^{h/2}\bigg<\prod_{i=1}^h\sum_{a=1}^N\frac{dx_i}{x_i-z_a}\bigg>_{N,\beta}^{\mathrm{(c)}},
\label{h_conn_diff}
\end{equation}
where $\left<\mathcal{O}\right>_{N,\beta}^{\mathrm{(c)}}$ denotes the connected part \cite{Eynard:2007kz} of $\left<\mathcal{O}\right>_{N}^{\beta}$ introduced in (\ref{expextation}).
In the 't Hooft limit (\ref{limit_beta}), $h$-point multi-linear differentials admit an asymptotic expansion
\begin{equation}
W_h(x_1,\ldots,x_h)=\sum_{g,\ell=0}^{\infty}\hbar^{2g-2+h+\ell}\gamma^{\ell}W^{(g,h)}_{\ell}(x_1,\ldots,x_h).
\label{W_diff_asy}
\end{equation}
\end{definition}
For the class of genus $0$ spectral curves of the form
\begin{align}
\begin{split}
&
y(x)^2=M(x)^2\sigma(x),
\\
&
\sigma(x)=(x-a)(x-b),\ \ \ \ M(x)=c\prod_{i=1}^f(x-\alpha_i)^{m_i},
\label{spectral_gen}
\end{split}
\end{align}
the $2$-point multi-linear differential $W^{(0,2)}_{0}(x_1,x_2)$ takes form
\begin{equation}
W^{(0,2)}_{0}(x_1,x_2)=B(x_1,x_2)-\frac{dx_1dx_2}{(x_1-x_2)^2},
\end{equation}
where the \textit{Bergman kernel} $B(x_1,x_2)$ is a bilinear differential 
\begin{equation}
B(x_1,x_2)=\frac{dx_1dx_2}{2(x_1-x_2)^2}\bigg[1+\frac{x_1x_2-\frac12(a+b)(x_1+x_2)+ab}{\sqrt{\sigma(x_1)\sigma(x_2)}}\bigg].
\label{Berg_one}
\end{equation}
An exact formula for $W^{(0,1)}_1(x)$ is also known \cite{Chekhov:2005rr, Chekhov:2006rq, Brini:2010fc},
\begin{equation}
W^{(0,1)}_1(x)=-\frac{dy(x)}{2y(x)}+\frac{dx}{2\sqrt{\sigma(x)}}\bigg[1+\sum_{i=1}^fm_i\bigg(1+\frac{\sqrt{\sigma(\alpha_i)}}{x-\alpha_i}\bigg)\bigg].
\label{W011_one}
\end{equation}
All other differentials $W^{(g,h)}_{\ell}(x_1,\ldots,x_h)$ in the asymptotic expansion (\ref{W_diff_asy}) can be determined recursively by means of the topological recursion, as stated in the theorem below.
\begin{theorem}[\hspace{-0.01em}\cite{Chekhov:2006rq, Chekhov:2010xj, Brini:2010fc}]
The differentials $W^{(g,h)}_{\ell}(x_H)$ for $(g,h,\ell)\neq(0,1,0)$, $(0,2,0)$, $(0,1,1)$ in the asymptotic expansion (\ref{W_diff_asy}) obey the $\beta$-deformed topological recursion
\begin{align}
&
\mathcal{W}^{(0,1)}_{0}(x)=0,\ \ \ \mathcal{W}^{(0,2)}_{0}(x_1,x_2)=W^{(0,2)}_{0}(x_1,x_2)+\frac{dx_1dx_2}{2(x_1-x_2)^2},\nonumber\\
&
\mathcal{W}^{(g,h)}_{\ell}(x_H)=W^{(g,h)}_{\ell}(x_H)\ \textrm{for}\ (g,h,\ell)\neq(0,1,0),\ (0,2,0),\nonumber\\
&
W^{(g,h+1)}_{\ell}(x,x_H)=\oint_{\mathcal{A}}\frac{1}{2\pi i}\frac{dS(x,z)}{y(z)dz}\bigg[W^{(g-1,h+2)}_{\ell}(z,z,x_H)\nonumber\\
&
\hspace{8em}+\sum_{k=0}^{g}\sum_{n=0}^{\ell}\sum_{\emptyset=J\subseteq H}\mathcal{W}^{(g-k,|J|+1)}_{\ell-n}(z,x_J)\mathcal{W}^{(k,|H|-|J|+1)}_{n}(z,x_{H \backslash J})\nonumber\\
&
\hspace{8em}+dz^2\frac{\partial}{\partial z}\frac{W^{(g,h+1)}_{\ell-1}(z,x_H)}{dz}\bigg],
\label{ref_loop_eq}
\end{align}
where $H=\{1,2,\ldots,h\}\supset J=\{i_1,i_2,\ldots,i_j\}$, $H\backslash J=\{i_{j+1},i_{j+2},\ldots,i_h\}$, and $\mathcal{A}$ is the counterclockwise contour around the branch cut $[a,b]$.
Here $dS(x_1,x_2)$ is the third type differential, which for the genus $0$ spectral curve (\ref{spectral_gen}) takes form
\begin{equation}
dS(x_1,x_2)=\frac{\sqrt{\sigma(x_2)}}{\sqrt{\sigma(x_1)}}\frac{dx_1}{x_1-x_2}.
\label{third_type_one}
\end{equation}
\end{theorem}

Our task in what follows is to determine generating functions of chord diagrams encoded in the free energy of the RNA matrix model (\ref{non_ori_rel1}). To this end we take advantage of the following trick \cite{ACPRS}. The potential of the RNA matrix model $V_{\mathrm{RNA}}(x)$ can be separated into the Gaussian part $V_{\mathrm{G}}(x)$ and the rational part $V_{\mathrm{rat}}(x)$
\begin{align}
V_{\mathrm{RNA}}(x)=V_{\mathrm{G}}(x)+V_{\mathrm{rat}}(x)+s,\ \ \ \
V_{\mathrm{G}}(x)=\frac{1}{2}x^2,
\ \ \ \
V_{\mathrm{rat}}(x)=-\frac{t^{-1}s}{t^{-1}-x}.
\label{Gauss_pot}
\end{align}
Adopting this separation into the definition of $Z_N^{\beta}(s,t)$ in (\ref{eigenvalue_integral_beta}), 
we find that the partition function of the RNA matrix model can be re-expressed as an expectation value in the $\beta$-deformed Gaussian model
\begin{align}
Z_N^{\beta}(s,t)
=\mathrm{e}^{-\frac{\sqrt{\beta}}{\hbar}sN}\bigg\langle
\exp\bigg[\frac{s\sqrt{\beta}}{t\hbar}\sum_{a=1}^N\frac{1}{t^{-1}-z_a}
\bigg]
\bigg\rangle_{N,\beta}^{\mathrm{G}},
\nonumber
\end{align}
where $\left<\mathcal{O}(\{z_a\})\right>_{N,\beta}^{\mathrm{G}}$ 
denotes $\left<\mathcal{O}(\{z_a\})\right>_{N}^{\beta}$ defined in (\ref{expextation}) with the Gaussian potential $V(x)=V_\mathrm{G}(x)$. It follows that the right hand side of this equation is given by a sum of  $h$-point multi-linear differentials with $x_i=t^{-1}$ ($i=1,\ldots,h$) in the Gaussian model.
\begin{proposition}\label{prop:Gauss_resolvent}
Let  $\omega_b(x_1,\ldots,x_b)$ denote the connected $b$-resolvent in the $\beta$-deformed Gaussian model, which we also refer to as the \textit{Gaussian $b$-resolvent},
\begin{equation}
\omega_b(x_1,\ldots,x_b)=\frac{W_h(x_1,\ldots,x_b)}{dx_1\cdots dx_b}=
\beta^{b/2}\bigg<\prod_{i=1}^b\sum_{a=1}^N\frac{1}{x_i-z_a}\bigg>_{N,\beta}^{\mathrm{(c)}}.
\label{b_conn_res_g}
\end{equation}
Then the free energy $F_N^{\beta}(s,t)$ of the $\beta$-deformed RNA matrix model is given by
\begin{equation}
F_N^{\beta}(s,t)=-\frac{\sqrt{\beta}}{\hbar}sN
+\sum_{b=1}^{\infty}\frac{s^b}{b!t^b\hbar^b}
\omega_b\big(t^{-1},\ldots,t^{-1}\big).
\end{equation}
\end{proposition}
Combining the above relation with the form of the free energy in (\ref{non_ori_rel1}), we obtain the key theorem for our enumeration of chord diagrams.
\begin{theorem}\label{theo:Gaussian_resol}
The coefficients $\widetilde{C}_{g,\ell,b}(w)$ in (\ref{non_ori_rel1}) are obtained directly from the principal specialization of the Gaussian $b$-resolvents
\begin{equation}
\omega_b(x,\ldots,x)=
\frac{1}{x^{b}}\sum_{g=0,\ell=0}^{\infty}(\mu^{-1}\hbar)^{2g-2+b+\ell}(-\gamma)^{\ell}\widetilde{C}_{g,\ell,b}(\mu x^{-2}).
\label{rna_g_res}
\end{equation}
\end{theorem}
On the basis of this theorem, one can compute the generating function $\widetilde{C}_{g,\ell,b}(w)$ of the numbers of non-oriented chord diagrams using the $\beta$-deformed topological recursion for the Gaussian model. We will present some details of such a computation in Section \ref{sec:top_rec}, and compare its results with other methods of enumeration.

Finally, in order to find the complete form of the asymptotic expansion, we also need to consider the \textit{unstable} part of the free energy, which is not determined by the topological recursion and must be computed independently. The unstable part consists of four terms: $F_{0,0}(s,t)$, $F_{0,1}(s,t)$, $F_{0,2}(s,t)$, and $F_{1,0}(s,t)$. The two terms $F_{0,0}(s,t)$ and $F_{1,0}(s,t)$ are the same as the terms $F_0(s,t)$ and $F_1(s,t)$  in the asymptotic expansion (\ref{Fg_asympt_g}) of the Hermitian matrix model. On the other hand, $F_{0,1}(\{r_n\})$ and $F_{0,2}(\{r_n\})$ can be computed using famous formulae, referred to respectively as the Dyson's formula and Wiegmann-Zabrodin formula, which for the genus $0$ spectral curve (\ref{spectral_gen}) take the following form.
\begin{theorem}\label{thm:unstable}
For the class of genus $0$ spectral curves (\ref{spectral_gen}) determined by
the general potential (\ref{g_pot}), the unstable parts 
$F_{0,1}(\{r_n\})$ and $F_{0,2}(\{r_n\})$ of the free energy are given by
\begin{align}
\frac{\partial}{\partial \mu}F_{0,1}(\{r_n\})&=1+\log |c|+\frac12\log\Big(\frac{a-b}{4}\Big)^2+\sum_{i=1}^fm_i\log\bigg[\frac12\Big(\alpha_i-\frac{a+b}{2}+\sqrt{\sigma(\alpha_i)}\Big)\bigg],
\label{RP2_one_cut0}
\\
F_{0,2}(\{r_n\})
&=-\frac{1}{2}\sum_{i=1}^fm_i\log\big(1-\mathfrak{s}_i^2\big)-\frac12\sum_{i,j=1}^fm_im_j\log\big(1-\mathfrak{s}_i\mathfrak{s}_j\big)
\nonumber\\
&\ \ \ \
+\frac{1}{24}\log\big|M(a)M(b)(a-b)^4\big|,
\label{Klein_one_cut0}
\end{align}
where $\mathfrak{s}_i$, $i=1,\ldots,f$ are defined as
\begin{equation}
\alpha_i(\mathfrak{s}_i)=\frac{a+b}{2}-\frac{a-b}{4}(\mathfrak{s}_i+\mathfrak{s}_i^{-1}),\ \ \ \ |\mathfrak{s}_i|<1.
\nonumber
\end{equation}
\end{theorem}

The formula (\ref{RP2_one_cut0}) for $F_{0,1}(\{r_n\})$ was found in \cite{Chekhov:2006rq, Chekhov:2010xj,Brini:2010fc}. On the other hand, the formula (\ref{Klein_one_cut0}) for $F_{0,2}(\{r_n\})$ is proven in appendix \ref{app:Klein_bottle}.

\subsection{A recursion relation from the quantum curve}

The main object that we consider in our second approach is the {\it wave-function} for the matrix model. The wave-function is the 1-point function defined as follows.
\begin{definition}\label{def:wave-function}
Let $\epsilon_{1,2}$ denote the parameters
\begin{equation}
\epsilon_1=-\beta^{1/2}g_s,\quad \epsilon_2=\beta^{-1/2}g_s,\quad 
g_s=2\hbar.
\label{matparam}
\end{equation}
For the general type potential (\ref{g_pot}),
the 1-point function $Z_{\alpha}(x;\{r_n\})$ ($\alpha=1,2$) 
\begin{align}
&
Z_{\alpha}(x;\{r_n\})
=\frac{1}{\mathrm{Vol}_N^{\beta}}
\int_{\mathbb{R}^N} \prod_{a=1}^Ndz_a 
\Delta(z)^{2\beta}\psi_{\alpha}(x)\e^{-\frac{\sqrt{\beta}}{\hbar}\sum_{a=1}^NV(z_a)},
\label{beta_matrix}
\end{align}
with the eigenvalue operator
\begin{align}
\psi_{\alpha}(x)=\prod_{a=1}^N(x-z_a)^{\frac{\epsilon_1}{\epsilon_{\alpha}}},
\nonumber
\end{align}
is referred to as the wave-function. For the RNA matrix model we denote the wave-function by $Z_{\alpha}(x;s,t)$.
\end{definition}

A prominent property of the wave-function \cite{ADKMV,ACDKV,MS,CHMS} is that it satisfies a partial differential equation, which is called the {\it quantum curve}\footnote{
The name ``quantum curve'' is also used for the ordinary differential equation for the wave function \cite{DHSV,DHS,Ho,GuSu,Z,MuSu,MSS,LMS,DM1,DMNPS,DoMa,Sch,DDM,DM2,No,GKMR,DM3,DKM,BE,DN,BCD}.
In this article, we also use this name for a partial differential equation which arises from the conformal field theoretical description of the matrix model \cite{Kostov}.
} 
or the  {\it time-dependent Schr\"odinger equation}.
\begin{proposition}[\hspace{-0.01em}\cite{ACDKV,MS}]\label{prop:ACDKV}
The wave-function $Z_{\alpha}(x;\{r_n\})$ satisfies the partial differential equation
\begin{align}
\left[-\epsilon_{\alpha}^2\frac{\partial^2}{\partial x^2}
-2\epsilon_{\alpha}V'(x)\frac{\partial}{\partial x}
+\hat{f}(x)
\right]Z_{\alpha}(x;\{r_n\})=0,
\label{loop_det}
\end{align}
where $\hat{f}(x)$ is the differential operator
\begin{align}
\hat{f}(x)=g_s^2\sum_{n=0}^Kx^n\partial_{(n)},\quad
\partial_{(n)}=\sum_{k=n+2}^{K+2}kr_k\frac{\partial}{\partial r_{k-n-2}},
\nonumber
\end{align}
and we denoted $\partial/\partial r_0=-N/(2\epsilon_2)$.
\end{proposition}

As its name suggests, the partial differential equation (\ref{loop_det}) is interpreted as the quantization of the spectral curve 
$\mathcal{C}$ of the matrix model.
Promoting the parameters $(x,\omega)$ in (\ref{spectral}) to the non-commutative operators $(\hat{x},\hat{p})$, such that
\begin{align}
\hat{p}=\epsilon_{\alpha}\frac{\partial}{\partial x},\qquad \hat{x}=x,\qquad
[\hat{p},\hat{x}]=\epsilon_{\alpha},
\nonumber
\end{align}
the partial differential equation (\ref{loop_det}) can be written in the form
\begin{align}
\widehat{A}Z_{\alpha}(x;\{r_n\})=0,\qquad \widehat{A} \equiv A(\hat{x},\hat{p}),
\nonumber
\end{align}
for an appropriate choice of $A(\hat{x},\hat{p})$. The operator $\widehat{A}$ is interpreted as a quantization of the spectral curve $\mathcal{C}$, and the equation (\ref{loop_det}) reduces to the defining equation of the spectral curve in the classical limit  \cite{ACDKV,MS}.

Note that, more generally, to a given spectral curve one may associate an infinite family of wave-functions and corresponding quantum curves, which take form of Virasoro singular vectors \cite{MS,CHMS}. In particular the quantum curves in (\ref{loop_det}), for both values of $\alpha=1,2$, correspond to singular vectors at level 2. In this work we do not consider quantum curves at levels higher than 2. 

In what follows we use the partial differential equation (\ref{loop_det}) as a tool to determine the partition function $Z_N^{\beta}(s,t)$ of the $\beta$-deformed RNA matrix model. The point is that 
the partition function $Z_N^{\beta}(\{r_n\})$ of the matrix model with 
the general type potential (\ref{g_pot}) appears as the leading order term 
in the expansion
of the wave-function $Z_{\alpha}(x;\{r_n\})$ around $x=\infty$ 
\begin{equation}
Z_{\alpha}(x;\{r_n\})=x^{\epsilon_1N/\epsilon_{\alpha}}\left(Z_N^{\beta}(\{r_n\})+\mathcal{O}(x^{-1})\right).
\nonumber
\end{equation}
Therefore, from the expansion of the wave-function with $\alpha=2$ for the RNA matrix model
\begin{align}
& 
Z_{\alpha=2}(x;s,t)=\exp\left[S(x,s,t)\right], 
\nonumber \\
&
S(x,s,t)=-\frac{2\mu}{\epsilon_{2}}\log x+
\sum_{b=0}^{\infty}\sum_{p=0}^{\infty}S_{b,p}(t)s^bx^{-p},
\nonumber
\end{align}
the free energy $F_N^{\beta}(s,t)$ can be written as
\begin{align}
F_N^{\beta}(s,t)=\sum_{b=0}^{\infty}S_{b,p=0}(t)s^b.
\nonumber
\end{align}
We use the partial differential equation (\ref{loop_det}) to determine each $S_{b,p}(t)$.
The main result of this approach is summarized in the following theorem.
\begin{theorem}\label{thm:phase_RNA} 
The partial differential equation for the phase function $S(x,s,t)$ 
for the RNA matrix model takes form
\begin{align}
&
\epsilon_2^2\left[\frac{\partial^2}{\partial x^2}S(x,s,t)+\left(\frac{\partial}{\partial x}S(x,s,t)\right)^2\right]+2\epsilon_2\left(x-\frac{st}{(1-tx)^2}\right)\frac{\partial}{\partial x}S(x,s,t)\nonumber\\
&
+4\mu\left(1-\frac{st^2(2-tx)}{(1-tx)^2}\right)+\frac{\epsilon_1\epsilon_2st^2(2-tx)}{(1-tx)^2}\frac{\partial}{\partial s}S(x,s,t)+\frac{\epsilon_1\epsilon_2t^3}{1-tx}\frac{\partial}{\partial t}S(x,s,t)=0.
\label{time_S_F}
\end{align}
This equation generates a hierarchy of differential equations for the coefficients $S_{b,p}(t)$ of the phase function, and the phase function is determined recursively with respect to the backbone number $b$.
\end{theorem}

As an independent verification of this algorithm, we checked iterative computations for $b=1,2,3,4,5$  up to $\mathcal{O}(t^{12})$, and confirmed that they agree with the results of the $\beta$-deformed topological recursion from the first approach.

\section{Enumeration of chord diagrams via the topological recursion} \label{sec:top_rec}

In the introduction, we extended the RNA matrix model proposed 
in  \cite{ACPRS} to the $\beta$-deformed RNA matrix model given by 
(\ref{eigenvalue_integral_beta}), that enumerates both orientable and non-orientable chord diagrams via Proposition \ref{prop:rna_beta}. 
In this section we enumerate orientable and non-orientable chord diagrams by the formalism of the $\beta$-deformed topological recursion (\ref{ref_loop_eq}).

\subsection{Enumeration of chord diagrams via Gaussian resolvents}  \label{subsec:rna_ga} 

In principle, by applying the $\beta$-deformed topological recursion formalism 
to the spectral curve (\ref{RNA_curve}) of the RNA matrix model, 
one can recursively compute the asymptotic expansion (\ref{free_energy_beta}) of the free energy $F_N^{\beta}(s,t)$. 
In \cite{ACPRS}, using the topological recursion \cite{Alexandrov:2003pj,Eynard:2004mh,Eynard:2007kz} for the $\beta=1$ (which does not involve $W^{(g,h)}_{\ell\ge 1}(x_H)$ terms) RNA matrix model (\ref{matrix_integral}), $F_{2}(s,t)=F_{2,0}(s,t)$ and $F_{3}(s,t)=F_{3,0}(s,t)$ were explicitly computed. However, because of the complicated form of the curve (\ref{RNA_curve}), an explicit computation of $W^{(g,h)}_{\ell\ge 1}(x_H)$ is not easy. In this section we consider instead the $\beta$-deformed Gaussian matrix model with the Gaussian potential $V_{\mathrm{G}}(x)$ in (\ref{Gauss_pot}).

Using Theorem \ref{theo:Gaussian_resol}, one can compute the coefficients $\widetilde{C}_{g,\ell,b}(w)$ in (\ref{non_ori_rel1}) of the free energy $F_N^{\beta}(s,t)$ from the Gaussian $b$-resolvents (\ref{b_conn_res_g}). For the Gaussian matrix model, the spectral curve takes form
\begin{equation}
y_{\mathrm{G}}(x)^2=x^2-4\mu.
\label{gauss_spc}
\end{equation}
In order to apply the $\beta$-deformed topological recursion (\ref{ref_loop_eq}) to this Gaussian curve,
it is convenient to introduce the Zhukovsky variable $\mathrm{z}$ as
\begin{equation}
x({\mathrm{z}})=\sqrt{\mu}({\mathrm{z}}+{\mathrm{z}}^{-1}).
\nonumber
\end{equation}
In this variable the branch points $x=\pm 2\sqrt{\mu}$ are mapped to ${\mathrm{z}}=\pm 1$.
To express the $\beta$-deformed topological recursion (\ref{ref_loop_eq}) in this Zhukovsky variable, we define
\begin{align}
\begin{split}
&
\widehat{y}_{\mathrm{G}}({\mathrm{z}})d{\mathrm{z}}=y_{\mathrm{G}}(x({\mathrm{z}}))dx
=\sqrt{\mu}({\mathrm{z}}-{\mathrm{z}}^{-1})d{\mathrm{z}},
\\
& d\widehat{S}({\mathrm{z}}_1,{\mathrm{z}}_2)=dS(x_1({\mathrm{z}}_1),x_2({\mathrm{z}}_2))=\frac{d{\mathrm{z}}_1}{{\mathrm{z}}_1-{\mathrm{z}}_2}-\frac{d{\mathrm{z}}_1}{{\mathrm{z}}_1-{\mathrm{z}}_2^{-1}},
\\
& \widehat{W}^{(g,h)}_{\ell}({\mathrm{z}}_1,\ldots,{\mathrm{z}}_h)=W^{(g,h)}_{\ell}(x_1({\mathrm{z}}_1),\ldots,x_h({\mathrm{z}}_h)).
\label{loop_var_ch_zhu}
\end{split}
\end{align}
We then obtain the $\beta$-deformed topological recursion (\ref{ref_loop_eq}) for the Gaussian spectral curve in the Zhukovsky variable. More generally, the $\beta$-deformed topological recursion for genus $0$ spectral curves (\ref{spectral_gen}) in the Zhukovsky is discussed in detail in \cite{Marchal:2011iu,MS}.

\begin{corollary}
The differentials $\widehat{W}^{(g,h)}_{\ell}(\mathrm{z}_H)$ for $(g,h,\ell)\neq(0,1,0)$, $(0,2,0)$, $(0,1,1)$ obey the $\beta$-deformed topological recursion in the Zhukovsky variable
\begin{equation}
\widehat{W}^{(g,h+1)}_{\ell}({\mathrm{z}},{\mathrm{z}}_H)=\oint_{\widetilde{\mathcal{A}}}\frac{1}{2\pi i}
\frac{d\widehat{S}({\mathrm{z}},\zeta)}{\widehat{y}_{\mathrm{G}}(\zeta)d\zeta}\mathrm{Rec}^{(g,h+1)}_{\ell}(\zeta,{\mathrm{z}}_H).
\label{ref_loop_eq_z}
\end{equation}
Here
\begin{align}
\mathrm{Rec}^{(g,h+1)}_{\ell}(\zeta,{\mathrm{z}}_H)&=
\widehat{W}^{(g-1,h+2)}_{\ell}(\zeta,\zeta,{\mathrm{z}}_H)\nonumber\\
&\ \ \
+\sum_{k=0}^{g}\sum_{n=0}^{\ell}\sum_{\emptyset=J\subseteq H}\widehat{\mathcal{W}}^{(g-k,|J|+1)}_{\ell-n}(\zeta,{\mathrm{z}}_J)\widehat{\mathcal{W}}^{(k,|H|-|J|+1)}_{n}(\zeta,{\mathrm{z}}_{H \backslash J})\nonumber\\
&\ \ \ 
+d\zeta^2\left[\frac{\partial}{\partial\zeta}
+\frac{\partial^2\zeta}{\partial{w}^2}\left(\frac{\partial{w}}{\partial\zeta}\right)^2\right]\frac{\widehat{W}^{(g,h+1)}_{\ell-1}(\zeta,{\mathrm{z}}_H)}{d\zeta},
\label{ref_loop_rec}
\end{align}
with ${w}=\sqrt{\mu}(\zeta+\zeta^{-1})$, $|\zeta|>1$, and
\begin{align}
&
\widehat{\mathcal{W}}^{(0,1)}_{0}({\mathrm{z}})=0,\ \ \ \widehat{\mathcal{W}}^{(0,2)}_{0}({\mathrm{z}}_1,{\mathrm{z}}_2)=\widehat{W}^{(0,2)}_{0}({\mathrm{z}}_1,{\mathrm{z}}_2)+\frac{({\mathrm{z}}_1^2-1)({\mathrm{z}}_2^2-1)d{\mathrm{z}}_1d{\mathrm{z}}_2}{2({\mathrm{z}}_1-{\mathrm{z}}_2)^2({\mathrm{z}}_1{\mathrm{z}}_2-1)^2},\nonumber\\
&
\widehat{\mathcal{W}}^{(g,h)}_{\ell}({\mathrm{z}}_H)=\widehat{W}^{(g,h)}_{\ell}({\mathrm{z}}_H)\ \textrm{for}\ (g,h,\ell)\neq(0,1,0),\ (0,2,0),
\end{align}
where $\widetilde{\mathcal{A}}$ is the contour surrounding the unit disk $|\zeta|=1$. For the Gaussian model the initial data of the recursion (\ref{Berg_one}) and (\ref{W011_one}) takes form
\begin{align}
&
\widehat{W}^{(0,2)}_{0}({\mathrm{z}}_1,{\mathrm{z}}_2)=\frac{d{\mathrm{z}}_1d{\mathrm{z}}_2}{({\mathrm{z}}_1{\mathrm{z}}_2-1)^2},\\
&
\widehat{W}^{(0,1)}_1({\mathrm{z}})=\bigg(\frac{1}{{\mathrm{z}}}-\frac{1}{2({\mathrm{z}}-1)}-\frac{1}{2({\mathrm{z}}+1)}\bigg)d{\mathrm{z}}.
\end{align}
\end{corollary}

From Theorem \ref{theo:Gaussian_resol} the coefficients $\widetilde{C}_{g,\ell,b}(w)$ of the free energy $F_N^{\beta}(s,t)$ can now be computed. For example we find,
for $b=1$,
\begin{align}
\begin{split}
&
\widetilde{C}_{0,0,1}(w)=\frac{1-\sqrt{1-4w}}{2w},\ \ \
\widetilde{C}_{0,1,1}(w)=\frac{1-\sqrt{1-4w}}{2(1-4w)},\ \ \
\widetilde{C}_{1,0,1}(w)=\frac{w^2}{(1-4w)^{5/2}},
\\
&
\widetilde{C}_{0,2,1}(w)=\frac{w\big(1+w-\sqrt{1-4w}\big)}{(1-4w)^{5/2}},\ \ \
\widetilde{C}_{1,1,1}(w)=\frac{w^2\big(1+30w-(1+6w)\sqrt{1-4w}\big)}{2(1-4w)^4},\\
&
\widetilde{C}_{0,3,1}(w)=\frac{5w^2\big(1+2w-(1+w)\sqrt{1-4w}\big)}{(1-4w)^4},
\nonumber
\end{split}
\end{align}
for $b=2$,
\begin{align}
\begin{split}
&
\widetilde{C}_{0,0,2}(w)=\frac{w}{(1-4w)^2},\ \ \
\widetilde{C}_{0,1,2}(w)=\frac{w\big(1+18w-(1+4w)\sqrt{1-4w}\big)}{2(1-4w)^{7/2}},
\\
&
\widetilde{C}_{1,0,2}(w)=\frac{w^3(21+20w)}{(1-4w)^5},\ \ \
\widetilde{C}_{0,2,2}(w)=\frac{w^2\big(8+98w+38w^2-(8+45w)\sqrt{1-4w}\big)}{(1-4w)^5},
\\
&
\widetilde{C}_{1,1,2}(w)=\frac{w^3\big(21+1462w+2700w^2-(21+376w+240w^2)\sqrt{1-4w}\big)}{2(1-4w)^{13/2}},
\\
&
\widetilde{C}_{0,3,2}(w)=\frac{w^3\big(117+1316w+1182w^2-(117+854w+292w^2)\sqrt{1-4w}\big)}{(1-4w)^{13/2}},
\nonumber
\end{split}
\end{align}
for $b=3$,
\begin{align}
\begin{split}
&
\widetilde{C}_{0,0,3}(w)=\frac{2w^2(3+4w)}{(1-4w)^{9/2}},
\\
&
\widetilde{C}_{0,1,3}(w)=
\frac{w^2\left(3+160w+354w^2-(3+50w+40w^2)\sqrt{1-4w}\right)}{(1-4w)^{6}},
\\
&
\widetilde{C}_{1,0,3}(w)=\frac{12w^4(45+207w+68w^2)}{(1-4w)^{15/2}},
\\
&
\widetilde{C}_{0,2,3}(w)=\frac{2w^3\big(58+1797w+5232w^2+1004w^3-(58+977w+1416w^2)\sqrt{1-4w}\big)}{(1-4w)^{15/2}}.
\nonumber
\end{split}
\end{align}
Then, by (\ref{gen_chord_non}) in Proposition \ref{prop:rna_beta}, we obtain the generating functions $C_{g,b}(w)$ for orientable chord diagrams and $C_{h,b}^{\mathsf r}(w)$ for non-oriented chord diagrams. For instance, we obtain
\begin{align}
\begin{split}
&C_{0,1}(w)=1+w+2w^2+5w^3+14w^4+42w^5+132w^6+429w^7+\mathcal{O}(w^8),
\\
&C_{1,1}(w)=w^2+10w^3+70w^4+420w^5+2310w^6+12012w^7+\mathcal{O}(w^8),
\\
&C_{0,2}(w)=w+8w^2+48w^3+256w^4+1280w^5+6144w^6+\mathcal{O}(w^7),
\\
&C_{1,2}(w)=21w^3+440w^4+5440w^5+51840w^6+421120w^7+\mathcal{O}(w^8),
\\
&C_{0,3}(w)=6w^2+116w^3+1332w^4+11880w^5+90948w^6+\mathcal{O}(w^7),
\\
&C_{1,3}(w)=540w^4+18684w^5+350736w^6+4779720w^7+\mathcal{O}(w^8),
\label{enu_orient_c}
\end{split}
\end{align}
and
\begin{align}
\begin{split}
&C_{1,1}^{\mathsf r}(w)=
w+5w^2+22w^3+93w^4+386w^5+1586w^6+6476w^7+\mathcal{O}(w^8),\\
&C_{2,1}^{\mathsf r}(w)=
5w^2+52w^3+374w^4+2290w^5+12798w^6+67424w^7+\mathcal{O}(w^8),\\
&C_{1,2}^{\mathsf r}(w)=
8w^2+117w^3+1084w^4+8119w^5+53640w^6+\mathcal{O}(w^7),\\
&C_{2,2}^{\mathsf r}(w)=
111w^3+2404w^4+30442w^5+295500w^6+\mathcal{O}(w^7),\\
&C_{1,3}^{\mathsf r}(w)=
116w^3+3204w^4+49248w^5+561782w^6+\mathcal{O}(w^7),\\
&C_{2,3}^{\mathsf r}(w)=
2952w^4+105300w^5+2021396w^6+\mathcal{O}(w^7).
\label{enu_non_orient_c}
\end{split}
\end{align}
We note that the generating function $C_{h,b}^{\mathsf r}(w)$ with an even cross-cap number $h=2g$ enumerates both orientable and non-orientable chord diagrams, and therefore non-orientable chord diagrams are enumerated by
\begin{equation}
C_{2g,b}^{\mathsf r}(w)-C_{g,b}(w).
\end{equation}

\subsection{The unstable part of the free energy} \label{subsec:ustable} 

For a matrix model with the general potential (\ref{g_pot}), the unstable coefficients $F_{0,0}(\{r_n\})$, $F_{0,1}(\{r_n\})$, $F_{1,0}(\{r_n\})$ and $F_{0,2}(\{r_n\})$ in the asymptotic expansion (\ref{free_energy_beta_g}) of the free energy $F_N^{\beta}(\{r_n\})$ must be computed separately. For the general potential (\ref{g_pot}) and the genus $0$ spectral curve (\ref{spectral_gen}), the coefficients
$F_{0,0}(\{r_n\})$ \cite{Brezin:1977sv, Marino:2004eq} and $F_{1,0}(\{r_n\})$ 
\cite{Ambjorn:1992jf,Ambjorn:1992gw} (see \cite{Akemann:1996zr,Chekhov:2004vx} for multi-cut solutions) are given by
\begin{align}
&
F_{0,0}(\{r_n\})=-\mu\int_{[a,b]}dz \rho(z)V(z)+\mu^2\int_{[a,b]^2}dz dz'\rho(z)\rho(z')\log|z-z'|,
\\
&
F_{1,0}(\{r_n\})=-\frac{1}{24}\log\big|M(a)M(b) (a-b)^4\big|,
\label{torus_formula}
\end{align}
where $\rho(z)=\lim_{N \to \infty}\frac{1}{N}\sum_{a=1}^{N}\delta(z-z_a)$ is the eigenvalue density given by
\begin{equation}
\rho(z)=\frac{1}{2\pi i \mu}\big(W^{(0,1)}_{0}(z-i\epsilon)-W^{(0,1)}_{0}(z+i\epsilon)\big)
=\frac{1}{2\pi i \mu}y(z),\ \ \ \ z\in [a,b].
\nonumber
\end{equation}
In \cite{ACPRS}, $F_{0,0}(s,t)$ and $F_{1,0}(s,t)$ for the RNA matrix model (\ref{matrix_integral}) were computed using the above formulae. 

Furthermore, the coefficients $F_{0,1}(s,t)$ and $F_{0,2}(s,t)$  for the genus $0$ spectral curve (\ref{spectral_gen}) can be computed using Theorem \ref{thm:unstable}. In particular, by Proposition \ref{prop:rna_beta} the generating function 
$C_{2,b}^{\mathsf r}(w)-C_{1,b}(w)$ for the numbers of chord diagrams with the topology of the Klein bottle takes form
\begin{equation}
F_{1,0}(\{r_n\})+F_{0,2}(\{r_n\})=-\frac{1}{2}\sum_{i=1}^fm_i\log\big(1-\mathfrak{s}_i^2\big)-\frac12\sum_{i,j=1}^fm_im_j\log\big(1-\mathfrak{s}_i\mathfrak{s}_j\big).
\end{equation}

Using the above formulae for the spectral curve (\ref{RNA_curve}) of the RNA matrix model, we determined the generating functions $C_{0,b}(w)$, $C_{1,b}(w)$, $C_{1,b}^{\mathsf r}(w)$, and $C_{2,b}^{\mathsf r}(w)$. We checked that the results coincide with the results obtained from 
the Gaussian $b$-resolvents discussed in the previous subsection. Compared with the method discussed in the previous subsection, the advantage of this method is that we find all order generating functions for the backbone number $b$.

\section{Enumeration of chord diagrams via quantum curve techniques} \label{sec:BPZ}
In this section we consider a recursive computation of the
numbers of chord diagrams, based on the quantum curve equation (\ref{loop_det}) for the wave-function of the RNA matrix model defined in Definition \ref{def:wave-function}.
For the $\beta$-deformed RNA matrix model the quantum curve equation reduces to a partial differential equation in three parameters $x$, $s$, and $t$.
We solve this partial differential equation recursively and obtain the generating function for the numbers of chord diagrams as the leading term in the expansion of the wave-function near $x\to \infty$. 

\subsection{Differential equation for the wave-function from the quantum curve} \label{subsec:quantum_curve}

The quantum curve for the $\beta$-deformed RNA matrix model is the key equation we take advantage of in this section. 
\begin{proposition}\label{prop:BPZ_RNA}
Let $Z_{\alpha}(x;s,t)$ $(\alpha=1,2)$ denote the wave-function for the $\beta$-deformed RNA matrix model
\begin{align}
&
Z_{\alpha}(x;s,t)
=\frac{1}{\mathrm{Vol}_N^{\beta}}
\int_{\mathbb{R}^N} \prod_{a=1}^Ndz_a
\Delta(z)^{2\beta}\psi_{\alpha}(x)\e^{-\frac{\sqrt{\beta}}{\hbar}\sum_{a=1}^NV_{\mathrm{RNA}}(z_a)},
\label{beta_wave_RNA}
\end{align}
where $\psi_{\alpha}(x)=\prod_{a=1}^N(x-z_a)^{\frac{\epsilon_1}{\epsilon_{\alpha}}}$ and the potential $V(x)=V_{\mathrm{RNA}}(x)$ is chosen as in equation (\ref{potential_RNA}). Then the partial differential equation (\ref{loop_det}) reduces to
\begin{align}
\begin{split}
&
\Bigg[-\left(\epsilon_\alpha\frac{\partial}{\partial x}\right)^2-2\epsilon_\alpha\left(x-\frac{st}{(1-tx)^2}\right)\frac{\partial}{\partial x}-4\mu\left(1-\frac{st^2(2-tx)}{(1-tx)^2}\right)
\\
&\hspace{9em}
-\frac{\epsilon_1\epsilon_2st^2(2-tx)}{(1-tx)^2}\frac{\partial}{\partial s}-\frac{\epsilon_1\epsilon_2t^3}{1-tx}\frac{\partial}{\partial t}\Bigg]
Z_{\alpha}(x;s,t)=0,
\label{BPZ_RNA1}
\end{split}
\end{align}
where  $\mu$ denotes the 't Hooft parameter 
$\mu=\beta^{1/2}\hbar N=-\epsilon_1N/2$.
\end{proposition}

\begin{proof}
The RNA matrix model has the potential (\ref{potential_RNA}), with the coefficients $r_n$ in (\ref{g_pot}) that take form
\begin{equation}
r_0=0,\ \ \ \
r_2=\frac{1}{2}-st^2,\ \ \ \
r_n=-st^n\;\; (n\ne 0,2).
\nonumber
\end{equation}
Adopting this choice of coefficients, 
the action of $\hat{f}(x)$ in (\ref{loop_det}) on $Z_{\alpha}(x;s,t)$ 
can be rewritten solely in terms of derivatives with respect to $s$ and $t$ 
\begin{align}
\hat{f}(x)Z_{\alpha}(x;s,t)
&=
-\frac{2}{\epsilon_2}g_s^2NZ_{\alpha}(x;s,t)
+g_s^2\sum_{n=0}^{\infty}x^nt^{n+3}\frac{\partial}{\partial
t}Z_{\alpha}(x;s,t)
\nonumber \\
&\ \ \ \
+g_s^2s\sum_{n=0}^{\infty}(n+2)x^nt^{n+2}
\bigg(
\frac{2}{\epsilon_2}N+
\frac{\partial}{\partial s}
\bigg)Z_{\alpha}(x;s,t)
\nonumber \\
&=
g_s^2\bigg[
-\frac{2}{\epsilon_2}N+
\frac{t^3}{1-tx}\frac{\partial}{\partial t}
+s\frac{t^2(2-tx)}{(1-tx)^2}\bigg(\frac{2N}{\epsilon_2}+\frac{\partial}{\partial
s}\bigg)
\bigg]Z_{\alpha}(x;s,t).
\nonumber
\end{align}
Using the relation between $N$ and $\mu$ stated in Proposition \ref{prop:BPZ_RNA}, one obtains the partial differential equation (\ref{BPZ_RNA1}).
\end{proof}

For simplicity, in the following we choose $\alpha=2$ 
and consider the wave-function $Z_{2}(x;s,t)$.
On the basis of Proposition \ref{prop:BPZ_RNA} we describe an algorithm to compute the free energy of the $\beta$-deformed RNA matrix model. To this end we consider the wave-function in the two following limits.

The first limit we consider is such that $x\to\infty$.
By the definition of the wave-function, 
the partition function $Z_N^{\beta}(s,t)$ is encoded in this limit as follows
\begin{equation}
Z_{2}(x;s,t)=x^{\epsilon_1N/\epsilon_{2}}\left(Z_N^{\beta}(s,t)+\mathcal{O}(x^{-1})\right).
\end{equation}
Equivalently, in the $x\to \infty$ limit the free energy $F_N^{\beta}(x;s,t)=\log Z_N^{\beta}(x;s,t)$ 
is found from the \textit{phase function} 
\begin{align}
S(x,s,t)=\log Z_{2}(x;s,t).
\label{phase_fn_rna_d}
\end{align}
Taking advantage of Proposition \ref{prop:BPZ_RNA}, we find the nonlinear partial differential equation (\ref{time_S_F}) for the phase function $S(x,s,t)$
in Theorem \ref{thm:phase_RNA}.

The second limit we consider is $s\to 0$.
In this limit the RNA matrix model reduces to the Gaussian matrix model, 
and we denote by $Z_{\alpha}^{\mathrm{G}}(x)$ the corresponding Gaussian wave-function 
\begin{align}
Z_{\alpha}^{\mathrm{G}}(x)
=\lim_{s\to 0}Z_{\alpha}(x;s,t)
=\frac{1}{\mathrm{Vol}_N^{\beta}}
\int_{\mathbb{R}^N} \prod_{a=1}^Ndz_a
\Delta(z)^{2\beta}\psi_{\alpha}(x)\e^{-\frac{\sqrt{\beta}}{2\hbar}\sum_{a=1}^Nz_a^2}.
\label{wave_Gauss}
\end{align}
This Gaussian wave-function obeys an ordinary differential equation
\begin{align}
&\bigg[
-\bigg(\epsilon_\alpha\frac{\partial}{\partial x}\bigg)^2
-2\epsilon_\alpha x \frac{\partial}{\partial x}-4\mu\bigg]
Z_{\alpha}^{\mathrm{G}}(x)=0,
\label{BPZ_Gauss}
\end{align}
which is obtained in the $s\to 0$ limit of equation (\ref{BPZ_RNA1}).
Here we denote the phase function of the Gaussian wave-function for $\alpha=2$ by
\begin{align}
S_0(x)=\log Z_{2}^{\mathrm{G}}(x).
\label{phase_Gauss}
\end{align}

In order to determine the free energy $F_N^{\beta}(s,t)=\log Z_N^{\beta}(s,t)$, we consider now the expansion of the phase function $S(x,s,t)$. First, we consider the expansion of the phase function with respect to the backbone parameter $s$
\begin{align}
S(x,s,t)=\sum_{b=0}^{\infty}S_b(x,t)s^b,
\label{backbone_expansion}
\end{align}
which has the following properties:
\begin{itemize}
\item The phase function $S_0(x)$ agrees with that of the Gaussian model
\begin{align}
S_0(x)=S_{0}(x,t).
\end{align}

\item The expansion of the phase function $S_b(x,t)$ around $x=\infty$ takes form
\begin{align}
S_b(x,t)=-\frac{2\mu}{\epsilon_2}\delta_{b,0}\log x+\sum_{p=0}^{\infty}S_{b,p}(t)x^{-p},
\label{sol_Fb_exp}
\end{align}
where an additional $\log x$ term for $b=0$ appears from the limit of $\psi_{\alpha}(x)$.
In particular $S_{0,p}\equiv S_{0,p}(t)$ do not depend on $t$.
\item The free energy $F_N^{\beta}(s,t)$ is obtained as the generating function of $S_{b,0}(t)$
\begin{align}
F_N^{\beta}(s,t)=\sum_{b=1}^{\infty}S_{b,0}(t)s^b.
\label{sol_Free}
\end{align}
\end{itemize}

Applying the expansion (\ref{backbone_expansion}) 
to the nonlinear partial differential equation (\ref{time_S_F}),
one finds a hierarchy of differential equations that determine the functions $S_{b,p}(t)$ recursively.
Although our main task is to determine the functions $S_{b,0}(t)$ in (\ref{sol_Free}),
we need the extra data of the higher order terms $S_{b,p\ge 1}(t)$ to 
determine $S_{b,0}(t)$.\footnote{It is not easy to reduce this recursion relation between $S_{b,p}(t)$ to a simple recursion relation involving $S_{b,0}(t)$ only.}
In the following we solve the equation (\ref{time_S_F}) systematically in four steps.

\subsection{Solving the recursion relations in four steps} \label{subsec:algorithm}
Now we solve the recursion relation for $S_{b,p}(t)$ in the following steps.
\begin{description}
\item[Step 1]  
Determine the hierarchy of differential equations by expanding the equation (\ref{time_S_F})  in the parameter $s$.
\item[Step 2] 
Solve the ordinary differential equation (\ref{BPZ_Gauss}) for the Gaussian wave-function.
\item[Step 3]  
Determine $S_1(x,t)$ iteratively by solving the differential equation to order $\mathcal{O}(s^1)$ in Step 1, with the initial data $S_0(x)$ obtained in Step 2.
\item[Step 4]
Repeat the same analysis for $S_{b}(x,t)$, by adopting the $S_{b'(\le b-1)}(x,t)$ as an input data.
\end{description}

\noindent{\textbf{Step 1: Hierarchy of differential equations for the phase function.}}
We determine the form of the hierarchy of differential equations by substituting the expansion (\ref{backbone_expansion}) in the differential equation (\ref{time_S_F}).
Picking up coefficients of $s^0, s^1$, and $s^b$ $(b\ge 2)$ respectively, 
we obtain nonlinear partial differential equations for $S_b(x,t)$
\begin{align}
\label{S_F_F0}
&
\epsilon_2^2\frac{\partial^2}{\partial x^2}S_0(x)+\epsilon_2^2\left(\frac{\partial}{\partial x}S_0(x)\right)^2+2\epsilon_2x\frac{\partial}{\partial x}S_0(x)+4\mu=0,\\
&
\epsilon_2^2\frac{\partial^2}{\partial x^2}S_1(x,t)+2\epsilon_2^2\frac{\partial}{\partial x}S_0(x)\frac{\partial}{\partial x}S_1(x,t)
+2\epsilon_2x\frac{\partial}{\partial x}S_1(x,t)-\frac{2\epsilon_2t}{(1-tx)^2}\frac{\partial}{\partial x}S_0(x)\nonumber\\
\label{S_F_F1}
&
-\frac{4\mu t^2(2-tx)}{(1-tx)^2}+\frac{\epsilon_1\epsilon_2t^2(2-tx)}{(1-tx)^2}S_1(x,t)+\frac{\epsilon_1\epsilon_2t^3}{1-tx}\frac{\partial}{\partial t}S_1(x,t)=0,
\end{align}
and
\begin{align}
&
\epsilon_2^2\frac{\partial^2}{\partial x^2}S_b(x,t)+\epsilon_2^2\sum_{a=0}^b\frac{\partial}{\partial x}S_a(x,t)\frac{\partial}{\partial x}S_{b-a}(x,t)
+2\epsilon_2x\frac{\partial}{\partial x}S_b(x,t)\nonumber\\
&
-\frac{2\epsilon_2t}{(1-tx)^2}\frac{\partial}{\partial x}S_{b-1}(x,t)+\frac{b\epsilon_1\epsilon_2t^2(2-tx)}{(1-tx)^2}S_b(x,t)+\frac{\epsilon_1\epsilon_2t^3}{1-tx}\frac{\partial}{\partial t}S_b(x,t)=0.
\label{S_F_Fb}
\end{align}
The first differential equation (\ref{S_F_F0}) for $S_0(x)$ is equivalent to the quantum curve equation (\ref{BPZ_Gauss}) for the Gaussian model, and we find that (\ref{S_F_F0})--(\ref{S_F_Fb}) can be solved successively for $S_{b}(x,t)$ ($b=1,2,3,\ldots$).

\vspace{0.2cm}
\noindent{\textbf{Step 2: The Gaussian phase function.}}
The Gaussian part of the wave-function is necessary as an input data in order to solve the equation (\ref{S_F_F1}).\footnote{Any solution of the ordinary differential equation (\ref{BPZ_Gauss}) 
can be expressed in terms of Hermite polynomials. }
Substituting the expansion (\ref{sol_Fb_exp}) for $S_0(x)$ in the differential equation (\ref{S_F_F0}), we obtain the recursion relation for the $t$-independent (as follows from the $t$-independence of $S_0(x)$) coefficients $S_{0,p}\equiv S_{0,p}(t)$.

\begin{proposition}\label{recursion_BPZ_Gauss}
The coefficients $S_{0,p}$ in the $1/x$ expansion of the Gaussian phase function $S_0(x)$ obey the recursion relation
\begin{align}
\begin{split}
&
S_{0,2p-3}=0,\ \ S_{0,2}=\frac{\mu}{2\epsilon_2}(2\mu+\epsilon_2),
\\
&
S_{0,2p}=\frac{1}{2p}\Big\{\epsilon_2(p-1)(2p-1)S_{0,2p-2}-4\mu(p-1)S_{0,2p-2}
\\
&\hspace{5em}
+2\epsilon_2\sum_{q=1}^{p-2}q(p-q-1)S_{0,2q}S_{0,2p-2q-2}\Big\},
\label{recursion_Gauss}
\end{split}
\end{align}
where $p$ is a positive integer with $p\ge 2$.
\end{proposition}
Solving this recursion relation iteratively, one finds the expansion
\begin{align}
S_0(x)=&-\frac{2\mu}{\epsilon_2}\log x+
\frac{\mu}{2\epsilon_2x^2}(2\mu+\epsilon_2)
+\frac{\mu}{8\epsilon_2x^4}(4\mu+3\epsilon_2)(2\mu+\epsilon_2)
\nonumber \\
&+\frac{\mu}{24\epsilon_2 x^6}(2\mu+\epsilon_2)(15\epsilon_2^2+34\mu\epsilon_2+20\mu^2)
\nonumber \\
&+\frac{\mu}{64\epsilon_2x^8}(2\mu+\epsilon_2)(105\epsilon_2^3+310\mu\epsilon_2^2+316\mu^2\epsilon_2+112\mu^3)
\nonumber \\
&+\frac{\mu}{160\epsilon_2x^{10}}(2\mu+\epsilon_2)(945\epsilon_2^4+3288\mu\epsilon_2^3+4424\mu^2\epsilon_2^2+2752\mu^3\epsilon_2+672\mu^4)
\nonumber \\
&+\mathcal{O}(x^{-12}).
\label{Gaussian_sol}
\end{align}

\noindent{\textbf{Step 3: The 1-backbone phase function $S_1(x,t)$}}.

In this step we continue our analysis of the 1-backbone phase function $S_1(x,t)$.
For this purpose, 
we expand equation (\ref{S_F_F1}) with respect to $x^{-1}$ 
and consider 
differential equations obtained for the coefficients of $x^1, x^0, x^{-1}$, and $x^{-p}$ $(p\ge 2)$.
\begin{corollary}\label{cor:S1}
The coefficients $S_{1,p}(t)$ ($p=0,1,\ldots$) obey the following differential equations in the parameter $t$:
\begin{align}
&
2S_{1,1}(t)+\epsilon_1t(\Theta_t+1)S_{1,0}(t)-\frac{4\mu t}{\epsilon_2}=0,
\label{hier_1_1}
\\
&
4tS_{1,2}(t)-4S_{1,1}(t)+\epsilon_1t^2(\Theta_t+1)S_{1,1}(t)-\epsilon_1t(\Theta_t+2)S_{1,0}(t)+\frac{8\mu t}{\epsilon_2}=0,
\label{hier_1_2}
\\
&
6t^2S_{1,3}(t)-8tS_{1,2}(t)+\epsilon_1t^3(\Theta_t+1)S_{1,2}(t)-\epsilon_1t^2(\Theta_t+2)S_{1,1}(t)\nonumber\\
&\hspace{6em}
-2\epsilon_2t^2S_{1,1}(t)-4\mu t^2S_{1,1}(t)+2S_{1,1}(t)-\frac{4\mu t}{\epsilon_2}=0,
\label{hier_1_3}
\end{align}
and
\begin{align}
&
2(p+2)t^2S_{1,p+2}(t)-4(p+1)tS_{1,p+1}(t)+\epsilon_1t^3(\Theta_t+1)S_{1,p+1}(t)
\nonumber \\
&
-\epsilon_1t^2(\Theta_t+2)S_{1,p}(t)
-p(p+1)\epsilon_2t^2S_{1,p}(t)-4p\mu t^2S_{1,p}(t)+2pS_{1,p}(t)
\nonumber \\
&
+2(p-1)p\epsilon_2tS_{1,p-1}(t)+8(p-1)\mu tS_{1,p-1}(t)
-(p-2)(p-1)\epsilon_2S_{1,p-2}(t)
\nonumber \\
&
-4(p-2)\mu S_{1,p-2}(t)-2(p-1)tS_{0,p-1}
-2\epsilon_2t^2\sum_{q=2}^{p-1}q(p-q)S_{0,q}S_{1,p-q}(t)
\nonumber\\
&
+4\epsilon_2t\sum_{q=3}^{p-1}(q-1)(p-q)S_{0,q-1}S_{1,p-q}(t)
-2\epsilon_2\sum_{q=4}^{p-1}(q-2)(p-q)S_{0,q-2}S_{1,p-q}(t)=0,
\label{hier_1_gen}
\end{align}
where $\Theta_t=t\partial / \partial t$, and $S_{0,p}$ are solutions of the recursion relation (\ref{recursion_Gauss}). 
\end{corollary}

In Step 2, $S_{0,p}$ have been already determined iteratively. Therefore, substituting the solution (\ref{Gaussian_sol}) in equations (\ref{hier_1_1})--(\ref{hier_1_gen}), we obtain differential equations for $S_{1,p}(t)$. Furthermore, the following condition follows from the recursive structure 
\begin{align}
S_{1,p}(t)=\sum_{k=1}^{\infty}S_{1,p,k}t^{k},\ \ \ \
S_{1,p,k}=0\ \ \textrm{for}\ \ p+k \ \ \textrm{odd},
\nonumber
\end{align}
whose implementation accelerates the iteration. We have implemented this iteration on a computer and found iteratively solutions for $S_{1,p}(t)$ ($p=0,1,2,\ldots$), which are summarized in appendix \ref{app:iterative}.

\noindent{\textbf{Step 4: The multi-backbone phase functions $S_b(x,t)$ ($b\ge 2$).}}

To extend our analysis to the multi-backbone phase functions $S_b(x,t)$ ($b\ge 2$)
we expand the differential equation (\ref{S_F_Fb}) with respect to the parameter $x$,
in the same way as in the previous steps.
\begin{corollary}\label{cor:Sb}
The coefficients $S_{b,p}(t)$ of the multi-backbone phase function $S_b(x,t)$
obey 
\begin{align}
&
2(p+2)t^2S_{b,p+2}(t)-4(p+1)tS_{b,p+1}(t)+\epsilon_1t^3(\Theta_t+b)S_{b,p+1}(t)
-\epsilon_1t^2(\Theta_t+2b)S_{b,p}(t)
\nonumber \\
&\ \
-p(p+1)\epsilon_2t^2S_{b,p}(t)-4p\mu t^2S_{b,p}(t)+2pS_{b,p}(t)
+2(p-1)p\epsilon_2tS_{b,p-1}(t)
\nonumber \\
&\ \
+8(p-1)\mu tS_{b,p-1}(t)-(p-2)(p-1)\epsilon_2S_{b,p-2}(t)
-4(p-2)\mu S_{b,p-2}(t)
\nonumber \\
&\ \
-2(p-1)tS_{b-1,p-1}(t)
-\epsilon_2t^2\sum_{a=0}^b\sum_{q=1}^{p-1}q(p-q)S_{a,q}(t)S_{b-a,p-q}(t)
\nonumber \\
&\ \
+2\epsilon_2t\sum_{a=0}^b\sum_{q=2}^{p-1}(q-1)(p-q)S_{a,q-1}(t)S_{b-a,p-q}(t)
\nonumber\\
&\ \
-\epsilon_2\sum_{a=0}^b\sum_{q=3}^{p-1}(q-2)(p-q)S_{a,q-2}(t)S_{b-a,p-q}(t)=0,
\label{hier_b_gen}
\end{align}
where $S_{b,p}(t)=0$ for $p\le -1$. 
\end{corollary}

The following conditions that follow from (\ref{hier_b_gen}) again accelerate the iteration
\begin{align}
S_{b,p}(t)=\sum_{k=b}^{\infty}S_{b,p,k}t^{k},
\ \ \ \
&
S_{b,0,k}=0\ \ \textrm{for}\ \ b\ge 2,\ \ k\le 2b-3,
\nonumber\\
&
S_{b,p,k}=0\ \ \textrm{for}\ \ p+k \ \ \textrm{odd}.
\nonumber
\end{align}
Programming this recursion on a computer, we determined the multi-backbone phase function $S_b(x,t)$ iteratively.  Computational results for $S_{b,p}(t)$ $b=2,3$ are summarized in appendix \ref{app:iterative}.

\subsection{The free energy}

Finally, we collect all $S_{b,0}(t)$ obtained in the above four steps together,
and substitute them into the equation (\ref{sol_Free}).
Rewriting $\epsilon_{\alpha}$ ($\alpha=1,2$) in (\ref{matparam})
in terms of the parameters $\hbar$ and  $\gamma=\beta^{1/2}-\beta^{-1/2}$, we obtain the free energy of the $\beta$-deformed RNA matrix model
\begin{align}
&
F^{\beta}_N(s,t)=\sum_{\ell=0}^{\infty}s^bF_{b}(\mu,\hbar,\gamma;t),
\nonumber
\\
&
F_{1}(\mu,\hbar,\gamma;t)=
\left(\frac{\mu^2}{\hbar^2}-\frac{\mu}{\hbar}\gamma\right)t^2
+\left(
 \frac{2 \mu^3}{\hbar^2}-\frac{5 \mu^2}{\hbar}\gamma
+\left(3\gamma^2+1\right)\mu
\right)t^4
\nonumber \\
&
\quad\quad\quad\quad\quad\quad\quad
+\biggl(\frac{5\mu^4}{\hbar^2}-\frac{22\mu^3}{\hbar}\gamma
+\left(
32\gamma^2+10\right)\mu^2
-\left(
15\gamma^3+13\gamma\right)\mu\hbar
\biggr)t^6
\nonumber \\
&
\quad\quad\quad\quad\quad\quad\quad
+\Biggl(
\frac{14\mu^5}{\hbar^2}-\frac{93\mu^4}{\hbar}\gamma+\left(234\gamma^2+70\right)\mu^3
+\left(52\gamma^3+43\gamma\right)\mu^2\hbar
\nonumber \\
&\quad\quad\quad\quad\quad\quad\quad\quad\quad
+\left(105\gamma^4+160\gamma^2+1\right)\mu\hbar^2
\Biggr)t^8
+\mathcal{O}(t^{10}),
\nonumber 
\end{align}
\begin{align}
&
F_{2}(\mu,\hbar,\gamma;t)=
\frac{\mu}{2\hbar^2}t^2
+\left(\frac{4\mu^2}{\hbar^2}-\frac{4\mu}{\hbar}\gamma\right)t^4
\nonumber \\
&\quad\quad\quad\quad\quad\quad\quad
+\left(\frac{24\mu^3}{\hbar^2}-\frac{117\mu^2}{2\hbar}\gamma+\frac{\mu}{2}
\left(69\gamma^2+21\right)
\right)t^6
\nonumber \\
&\quad\quad\quad\quad\quad\quad\quad
+\Biggl(
\frac{128\mu^4}{\hbar^2}-\frac{542\mu^3}{\hbar}\gamma
+\mu^2\left(762\gamma^2+220\right)
-\mu\hbar\left(348\gamma^3+282\gamma\right)
\Biggr)t^8
\nonumber \\
&\quad\quad\quad\quad\quad\quad\quad
+\mathcal{O}(t^{10}),
\nonumber 
\\
&F_{3}(\mu,\hbar,\gamma;t)=
\frac{\mu}{\hbar^2}t^4
+\left(\frac{58\mu^2}{3\hbar^2}-\frac{58\mu}{3}\gamma\right)t^6
\nonumber \\
&\quad\quad\quad\quad\quad\quad\quad
+\left(\frac{222\mu^3}{\hbar^2}-\frac{534\mu^2}{\hbar}\gamma+\mu\left(
312\gamma^2+90\right)\right)t^8
+\mathcal{O}(t^{10}).
\nonumber
\end{align}
From (\ref{free_energy_beta}) and (\ref{non_ori_rel1}) we obtain the coefficients $\widetilde{C}_{g,\ell,b}(w)$
in $F_{b}(\mu,\hbar,\gamma;t)$ via the formula
\begin{align}
F_{b}(\mu,\hbar,\gamma;t)
=\sum_{
g,\ell=0
}^{\infty}\frac{(-\gamma)^{\ell}\hbar^{2g-2+\ell}}{\mu^{b-2+2g+\ell}b!}\widetilde{C}_{g,\ell,b}(\mu t^2)
-\mu \hbar^{-2} \delta_{b,1}.
\end{align}
Taking into account the relations (\ref{gen_chord_non})
for $C_{g,b}(w)$ and $C_{h,b}^{\mathsf r}(w)$
\begin{align}
C_{g,b}(w)=\widetilde{C}_{g,0,b}(w),
\ \ \ \
C_{h,b}^{\mathsf r}(w)=\sum_{
\substack{
g,\ell=0\\
2g+\ell=h}
}^{\infty}2^g\widetilde{C}_{g,\ell,b}(w),
\nonumber
\end{align}
we find complete agreement with the computational results (\ref{enu_orient_c}) 
and (\ref{enu_non_orient_c}) obtained from the $\beta$-deformed topological recursion via Gaussian resolvents.


\appendix

\section{Enumeration of non-orientable fatgraphs with $k=3$}
\label{app:non_ori_fatgraph}

For the case of one backbone (i.e. one vertex) and $k=3$ chords, we list the 
non-oriented fatgraphs with $\chi=1$ in Figure \ref{non_orient_chi1} and $\chi=0$ in Figure \ref{non_orient_chi0}.

\begin{figure}[h]
\begin{center}
   \includegraphics[width=120mm,clip]{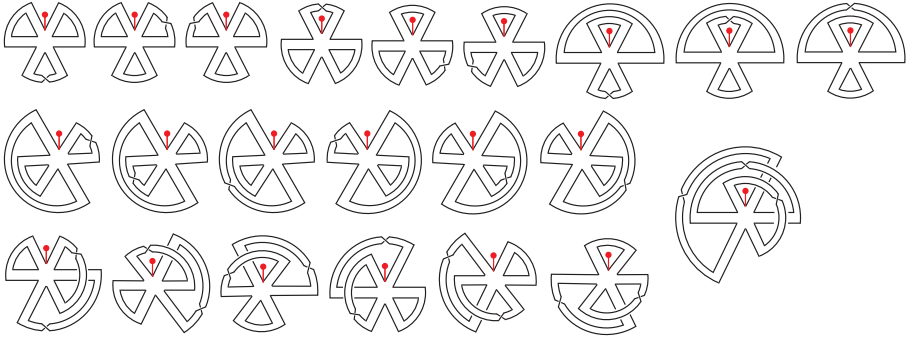}
\end{center}
\caption{\label{non_orient_chi1} Non-oriented tailed fatgraphs with $b=1$, $k=3$
 and $\chi=1$. The total number of graphs is 22.}
\end{figure}

\begin{figure}[h]
\begin{center}
   \includegraphics[width=120mm,clip]{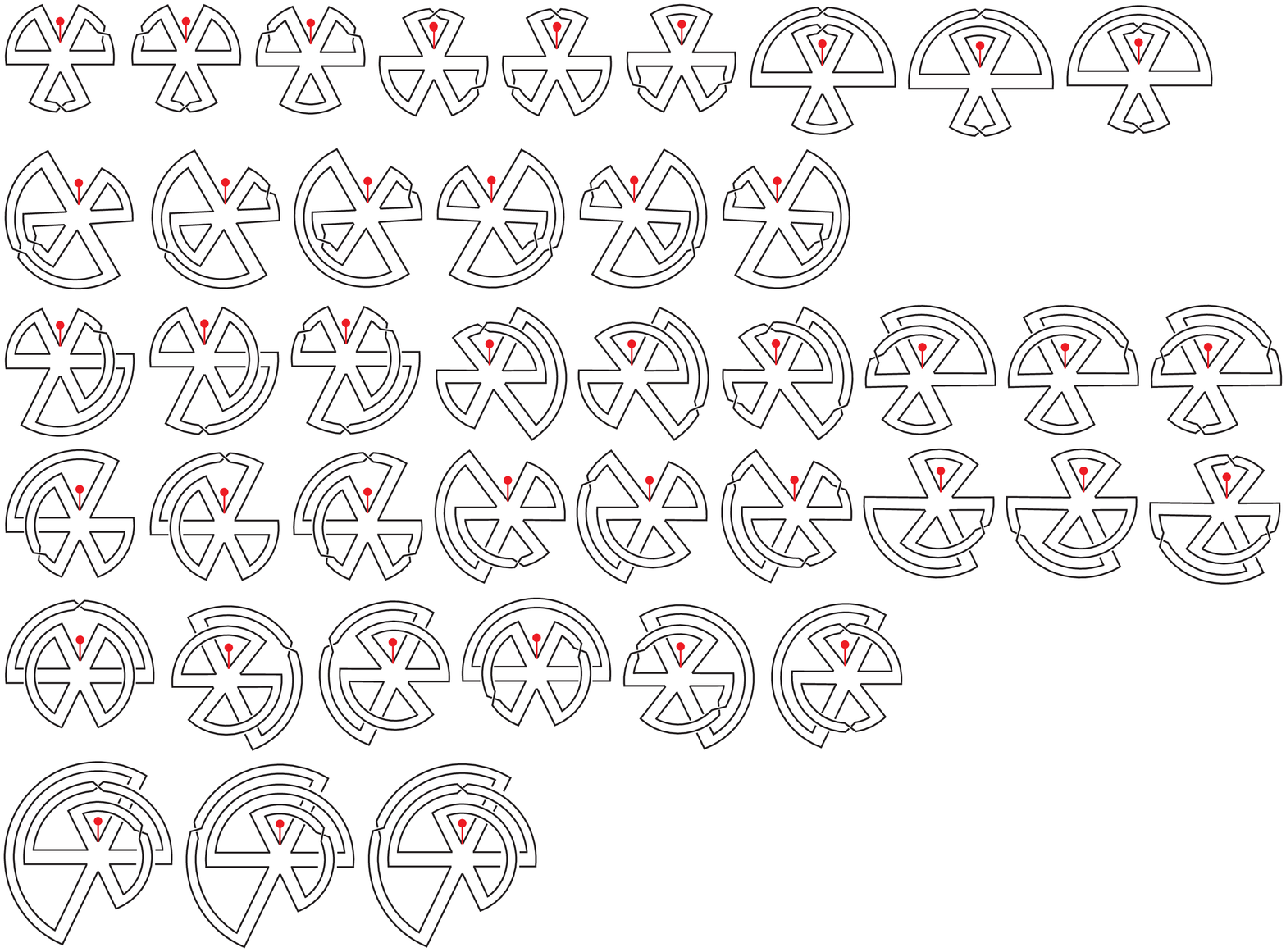}
\end{center}
\caption{\label{non_orient_chi0} Non-oriented tailed fatgraphs with $b=1$, $k=3$
 and $\chi=0$. The total number of graphs is 42.}
\end{figure}

Both of these numbers of graphs agree with the number computed from the cut-and-join
method and the time-dependent Schr\"odinger equation.

\section{The free energy $F_{0,2}$ in 1-cut $\beta$-deformed models}\label{app:Klein_bottle}

In this appendix, we determine the unstable term $F_{0,2}(\{r_n\})$ in the free energy (\ref{free_energy_beta_g}) for the $\beta$-deformed eigenvalue integral (\ref{eigenvalue_integral_beta}) with the general potential (\ref{g_pot})
\begin{equation}
F_{0,2}(\{r_n\})=F_{0,2}^I(\{r_n\})+F_{0,2}^A(\{r_n\}).
\end{equation}
We consider genus $\mathsf{s}-1$ spectral curves of the form
\begin{align}
\begin{split}
&
y(x)^2=M(x)^2\sigma(x),
\\
&
\sigma(x)=\prod_{i=1}^{2\mathsf{s}}(x-{\mathsf{q}}_i),\ \ \ \ 
M(x)=c\prod_{i=1}^f(x-\alpha_i)^{m_i},
\label{spec_gen}
\end{split}
\end{align}
whose form is determined by the choice of the potential.
The terms $F_{0,2}^I(\{r_n\})$ and $F_{0,2}^A(\{r_n\})$ are given by \cite{Chekhov:2006rq, Chekhov:2010xj}
\begin{align}
\label{Klein_I}
&
F_{0,2}^I(\{r_n\})=
-\frac{1}{8\pi^2}\oint_{\mathcal{A}}\frac{dy(z')}{y(z')}
\int_DdS(z,z')\log |y(z)|,\\
\label{Klein_A}
&
F_{0,2}^A(\{r_n\})=-\frac{1}{12}\log \bigg|\prod_{i=1}^{2\mathsf{s}}M({\mathsf{q}}_i)\cdot \prod_{1\le i<j\le 2\mathsf{s}}({\mathsf{q}}_i-{\mathsf{q}}_j)\bigg|,
\end{align}
where $\mathcal{A}=\bigcup_{i=1}^{\mathsf{s}}\mathcal{A}_{i}$ is the contour around the branch cut $D=\bigcup_{i=1}^{\mathsf{s}}D_i$, $D_i=[{\mathsf{q}}_{2i-1}, {\mathsf{q}}_{2i}]$. Then $dS(x_1,x_2)$ is the third type differential, which is a 1-form in $x_1$ and a multivalued function of $x_2$, determined by the conditions
\begin{align}
&
\bullet\ dS(x_1,x_2)\mathop{\sim}_{x_1 \to x_2} \frac{dx_1}{x_1-x_2}+\mathrm{reg.},\ \ \ \
\bullet\ dS(x_1,x_2)\mathop{\sim}_{x_1 \to \overline{x}_2} -\frac{dx_1}{x_1-x_2}+\mathrm{reg.},\nonumber\\
&
\bullet\ \oint_{x_2 \in{\mathcal{A}}_{i}}dS(x_1,x_2)=0,\ \ \ \ i=1,\ldots,\mathsf{s}-1.
\label{third_type_def}
\end{align}
Here $\overline{x}$ is the conjugate point of a point $x$ on the spectral curve 
(\ref{spec_gen}), such that
\begin{equation}
\sqrt{\sigma(x)}=-\sqrt{\sigma(\overline{x})},\ \ \ \
M(x)=M(\overline{x}).
\end{equation}

In the following we consider the $\mathsf{s}=1$ case in (\ref{spec_gen})
\begin{equation}
\sigma(x)=(x-a)(x-b),\ \ \ \ 
a<b,
\end{equation}
and by applying the method used to derive the unstable term $F_{0,2}(\{r_n\})$ in a 1-cut matrix model with a hard edge \cite{Borot:2010tr}, we prove the formula (\ref{Klein_one_cut0}).
\begin{proposition}
For the above genus $0$ spectral curve, the unstable term $F_{0,2}(\{r_n\})$ takes form
\begin{align}
\begin{split}
F_{0,2}(\{r_n\})&
=-\frac{1}{2}\sum_{i=1}^fm_i\log\big(1-\mathfrak{s}_i^2\big)-\frac12\sum_{i,j=1}^fm_im_j\log\big(1-\mathfrak{s}_i\mathfrak{s}_j\big)
\\
&\ \ \ \ 
+\frac{1}{24}\log\big|M(a)M(b)(a-b)^4\big|,
\label{Klein_formula}
\end{split}
\end{align}
where $\mathfrak{s}_i$, $i=1,\ldots,f$ are defined by
\begin{equation}
\alpha_i(\mathfrak{s}_i)=\frac{a+b}{2}-\frac{a-b}{4}(\mathfrak{s}_i+\mathfrak{s}_i^{-1}),\ \ \ \ |\mathfrak{s}_i|<1.
\end{equation}
\end{proposition}

\begin{proof}
In the above genus $0$ case the third type differential $dS(x_1,x_2)$ is given by
\begin{equation}
dS(x_1,x_2)=\frac{\sqrt{\sigma(x_2)}}{\sqrt{\sigma(x_1)}}\frac{dx_1}{x_1-x_2}.
\label{third_type}
\end{equation}
We introduce the Zhukovsky variable ${\mathrm{z}}$ by
\begin{equation}
x({\mathrm{z}})=\frac{a+b}{2}-\frac{a-b}{4}({\mathrm{z}}+{\mathrm{z}}^{-1}).
\label{Zhukovsky}
\end{equation}
Then the branch points $x=a, b$ are mapped to ${\mathrm{z}}=-1, +1$, and 
the first and second sheet of the spectral curve are mapped to the regions $|{\mathrm{z}}|\ge 1$ and $|{\mathrm{z}}|\le 1$, respectively. Under this map we obtain
\begin{align}
\sqrt{\sigma(x)}=\frac{b-a}{4}({\mathrm{z}}-{\mathrm{z}}^{-1}),\ \ \ \
\frac{dx}{\sqrt{\sigma(x)}}=\frac{d{\mathrm{z}}}{{\mathrm{z}}},\ \ \ \
\frac{\sqrt{\sigma(x_2)}}{x_1-x_2}=\frac{{\mathrm{z}}_1}{{\mathrm{z}}_1-{\mathrm{z}}_2}-\frac{{\mathrm{z}}_1}{{\mathrm{z}}_1-{\mathrm{z}}_2^{-1}},
\nonumber
\end{align}
and the third type differential (\ref{third_type}) is rewritten as
\begin{equation}
dS(x_1,x_2)=\frac{d{\mathrm{z}}_1}{{\mathrm{z}}_1-{\mathrm{z}}_2}-\frac{d{\mathrm{z}}_1}{{\mathrm{z}}_1-{\mathrm{z}}_2^{-1}}.
\label{third_type_z}
\end{equation}
Under the map (\ref{Zhukovsky}), the zeros or poles $\alpha_i$ of the moment function $M(x)$ on the spectral curve (\ref{spec_gen}) are mapped to $2f$ points $\mathfrak{s}_i^{\pm 1}$, $i=1,\ldots,f$,
\begin{equation}
\alpha_i(\mathfrak{s}_i)=\frac{a+b}{2}-\frac{a-b}{4}(\mathfrak{s}_i+\mathfrak{s}_i^{-1}),
\nonumber
\end{equation}
and we obtain
\begin{equation}
x-\alpha_i=\frac{b-a}{4}\frac{({\mathrm{z}}-\mathfrak{s}_i)({\mathrm{z}}-\mathfrak{s}_i^{-1})}{{\mathrm{z}}}.
\nonumber
\end{equation}
Without loss of generality, in this proof we can assume
\begin{equation}
|\mathfrak{s}_i|>1.
\nonumber
\end{equation}

First, let us consider (\ref{Klein_I}). Since the variable $z$ of the integrand is on the branch cut $D=[a,b]$, we can put ${\mathrm{z}}=e^{i\theta}$ in the Zhukovsky variable, and by (\ref{third_type_z})
\begin{equation}
dS(z,z')=-id\theta\bigg[\frac{{\mathrm{z}}}{{\mathrm{z}}'}\frac{1}{1-\frac{{\mathrm{z}}}{{\mathrm{z}}'}}+\frac{1}{1-\frac{1}{{\mathrm{z}}{\mathrm{z}}'}}\bigg]=-id\theta\bigg[1+\sum_{k=1}^{\infty}\frac{2}{{\mathrm{z}}'^k}\cos k\theta\bigg]
\nonumber
\end{equation}
is obtained, where we have used $|{\mathrm{z}}'|>|{\mathrm{z}}|=1$. Then we obtain
\begin{align}
\widehat{F}_{0,2}^I(z')&:=
\int_DdS(z,z')\log |y(z)|
\nonumber\\
&\;
=\int_0^{\pi}id\theta\bigg[1+\sum_{k=1}^{\infty}\frac{2}{{\mathrm{z}}'^k}\cos k\theta\bigg]
\bigg[\log \bigg|c\frac{b-a}{4}\prod_{i=1}^f\Big(\frac{b-a}{4}\Big)^{m_i}\bigg|+\log (2\sin\theta)
\nonumber\\
&\hspace{15em}
+\log\prod_{i=1}^f|e^{i\theta}-\mathfrak{s}_i|^{m_i}|e^{i\theta}-\mathfrak{s}_i^{-1}|^{m_i}\bigg]\nonumber\\
&\;
=
i\pi\log \bigg|c\frac{b-a}{4}\prod_{i=1}^f\Big(\frac{b-a}{4}\Big)^{m_i}\bigg|
-i\pi\sum_{k=1}^{\infty}\frac{1}{k{\mathrm{z}}_2^{2k}}\nonumber\\
&\hspace{1em}
+i\sum_{i=1}^fm_i\int_0^{\pi}d\theta\bigg[1+\sum_{k=1}^{\infty}\frac{2}{{\mathrm{z}}_2^k}\cos k\theta\bigg]\log|e^{i\theta}-\mathfrak{s}_i||e^{i\theta}-\mathfrak{s}_i^{-1}|,
\label{FI02_x_2_a}
\end{align}
where we have used that
\begin{align}
\begin{split}
&
\int_0^{\pi}d\theta\log (2\sin\theta)=0,\\
&
\int_0^{\pi}d\theta \cos k\theta \log\sin\theta=
\begin{cases}
-\frac{\pi}{k}& \textrm{if $k$ is a nonzero even integer},\\
0& \textrm{if $k$ is an odd integer}.
\end{cases}
\nonumber
\end{split}
\end{align}
By
\begin{align}
\begin{split}
&
\log|e^{i\theta}-\mathfrak{s}_i|=\log |\mathfrak{s}_i|-\frac12\sum_{k=1}^{\infty}\frac{1}{k}\Big(\frac{1}{\mathfrak{s}_i^k}+\frac{1}{\overline{\mathfrak{s}}_i^k}\Big)\cos k\theta
-\frac{i}{2}\sum_{k=1}^{\infty}\frac{1}{k}\Big(\frac{1}{\mathfrak{s}_i^k}-\frac{1}{\overline{\mathfrak{s}}_i^k}\Big)\sin k\theta,\nonumber\\
&
\log|e^{i\theta}-\mathfrak{s}_i^{-1}|=
-\frac12\sum_{k=1}^{\infty}\frac{1}{k}\Big(\frac{1}{\mathfrak{s}_i^k}+\frac{1}{\overline{\mathfrak{s}}_i^k}\Big)\cos k\theta
+\frac{i}{2}\sum_{k=1}^{\infty}\frac{1}{k}\Big(\frac{1}{\mathfrak{s}_i^k}-\frac{1}{\overline{\mathfrak{s}}_i^k}\Big)\sin k\theta,
\nonumber
\end{split}
\end{align}
(\ref{FI02_x_2_a}) is written as
\begin{equation}
\widehat{F}_{0,2}^I(z)=i\pi
\bigg[\log \bigg|c\frac{b-a}{4}\prod_{i=1}^f\Big(\frac{b-a}{4}\mathfrak{s}_i\Big)^{m_i}\bigg|
-\sum_{k=1}^{\infty}\frac{1}{k{\mathrm{z}}^{2k}}
-\sum_{i=1}^fm_i\sum_{k=1}^{\infty}\frac{2}{k{\mathrm{z}}^k\mathfrak{s}_i^k}\bigg].
\label{FI02_x_b}
\end{equation}
Here we have used
\begin{equation}
\int_0^{\pi}d\theta\cos k\theta \cos \ell \theta=\frac{\pi}{2}\delta_{k,\ell},
\nonumber
\end{equation}
and the fact that for an arbitrary $\overline{\mathfrak{s}}_i$, there exists an $\mathfrak{s}_j$ such that $\overline{\mathfrak{s}}_i=\mathfrak{s}_j$. By
\begin{equation}
\frac{dy(z)}{y(z)}=\bigg[\frac{1}{{\mathrm{z}}-1}+\frac{1}{{\mathrm{z}}+1}-\frac{1}{{\mathrm{z}}}+\sum_{i=1}^fm_i\Big(\frac{1}{{\mathrm{z}}-\mathfrak{s}_i}+\frac{1}{{\mathrm{z}}-\mathfrak{s}_i^{-1}}-\frac{1}{{\mathrm{z}}}\Big)\bigg]d{\mathrm{z}},
\nonumber
\end{equation}
and by (\ref{FI02_x_b}), (\ref{Klein_I}) is rewritten as
\begin{align}
F_{0,2}^I(\{r_n\})&
=\frac14\oint_{\widetilde{\mathcal{A}}}\frac{d{\mathrm{z}}}{2\pi i}\bigg[\frac{1}{{\mathrm{z}}-1}+\frac{1}{{\mathrm{z}}+1}-\frac{1}{{\mathrm{z}}}+\sum_{i=1}^fm_i\Big(\frac{1}{{\mathrm{z}}-\mathfrak{s}_i}+\frac{1}{{\mathrm{z}}-\mathfrak{s}_i^{-1}}-\frac{1}{{\mathrm{z}}}\Big)\bigg]\nonumber\\
&\hspace{1.5em}\times
\bigg[\log \bigg|c\frac{b-a}{4}\prod_{i=1}^f\Big(\frac{b-a}{4}\mathfrak{s}_i\Big)^{m_i}\bigg|
-\sum_{k=1}^{\infty}\frac{1}{k{\mathrm{z}}^{2k}}
-\sum_{i=1}^fm_i\sum_{k=1}^{\infty}\frac{2}{k{\mathrm{z}}^k\mathfrak{s}_i^k}\bigg],
\nonumber
\end{align}
where $\widetilde{\mathcal{A}}$ is the contour around the unit disk $|{\mathrm{z}}|=1$. This contour can be changed as
\begin{equation}
\oint_{\widetilde{\mathcal{A}}}\ \longrightarrow \ -\sum_{i=1}^f\oint_{\widetilde{\mathcal{A}}_i}-\oint_{\widetilde{\mathcal{A}}_{\infty}},
\nonumber
\end{equation}
where $\widetilde{\mathcal{A}}_i$ and $\widetilde{\mathcal{A}}_{\infty}$ are the contours around $\mathfrak{s}_i$ and infinity, respectively. Then by
\begin{equation}
\sum_{k=1}^{\infty}\frac{u^k}{k}=-\log(1-u),\ \ \ \ |u|<1,
\nonumber
\end{equation}
we obtain
\begin{align}
\begin{split}
F_{0,2}^I(\{r_n\})&=-\frac12\sum_{i=1}^fm_i\log\Big(1-\frac{1}{\mathfrak{s}_i^2}\Big)-\frac12\sum_{i,j=1}^fm_im_j\log\Big(1-\frac{1}{\mathfrak{s}_i\mathfrak{s}_j}\Big)
\\
&\ \ \ \ 
+\frac18\log\bigg|M(a)M(b)\Big(\frac{a-b}{4}\Big)^2\bigg|,
\label{Klein_I_res}
\end{split}
\end{align}
where we used
\begin{equation}
\log\bigg|c\prod_{i=1}^f\Big(\frac{b-a}{4}\mathfrak{s}_i\Big)^{m_i}\bigg|=
-\sum_{i=1}^fm_i\log\Big(1-\frac{1}{\mathfrak{s}_i^2}\Big)+\frac12\log\big|M(a)M(b)\big|.
\nonumber
\end{equation}
With (\ref{Klein_A}) for the genus $0$ case, after changing $\mathfrak{s}_i$ to $\mathfrak{s}_i^{-1}$, we finally obtain the formula (\ref{Klein_formula}),
where we ignored the constant term $-\log\sqrt{2}$.
\end{proof}

\section{An iterative solution using the quantum curve}\label{app:iterative}

\noindent{\textbf{\underline{The 1-backbone phase function $S_1(x,t)$}}}\\
The solutions $S_{1,p}(t)$ ($p=0,1,2$) of the equations (\ref{hier_1_1}) -- (\ref{hier_1_gen})
are found successively as follows
\begin{align}
S_{1,0}(t)&=-t^3\frac{2\mu (2\mu+\epsilon_1+\epsilon_2)}{\epsilon_2\epsilon_2}
-t^4\frac{\mu}{\epsilon_1\epsilon_2}
\big(5\epsilon_1\epsilon_2+8\mu^2+10\mu(\epsilon_1+\epsilon_2)+3(\epsilon_1^2+\epsilon_2^2)\big)
\nonumber \\
&\ \ \
-t^6\frac{\mu}{2\epsilon_1\epsilon_2}\bigl(40\mu^3+88\mu^2(\epsilon_1+\epsilon_2)+64\mu(\epsilon_1^2+\epsilon_2^2)+108\mu\epsilon_1\epsilon_2
+15(\epsilon_1^3+\epsilon_2^3)
\nonumber \\
&\qquad\quad\quad\quad\quad
+32\epsilon_1\epsilon_2(\epsilon_1+\epsilon_1)\bigr)
\nonumber \\
&\ \ \
-t^8\frac{\mu}{4\epsilon_1\epsilon_2}\bigl(224\mu^4+744\mu^3(\epsilon_1+\epsilon_2)
+936\mu^2(\epsilon_1^2+\epsilon_2^2)+1592\mu^2\epsilon_1\epsilon_2
\nonumber \\
&\qquad\quad\quad\quad\quad
+520\mu(\epsilon_1^3+\epsilon_2^3)
+1130\mu\epsilon_1\epsilon_2(\epsilon_1+\epsilon_2)
+260\epsilon_1\epsilon_1(\epsilon_1^2+\epsilon_2^2)
\nonumber \\
&\qquad\quad\quad\quad\quad
+331\epsilon_1^2\epsilon_2^2
+105(\epsilon_1^4+\epsilon_2^4)\bigr)+\mathcal{O}(t^{10}),
\nonumber
\end{align}
\begin{align}
S_{1,1}(t)&=t\frac{2\mu}{\epsilon_2}
+t^3\frac{3\mu(2\mu+\epsilon_1+\epsilon_2)}{\epsilon_2}
+t^5\frac{5\mu}{2\epsilon_2}
\big(8\mu^2+10\mu(\epsilon_1+\epsilon_2)
+3(\epsilon_1^2+\epsilon_2^2)+5\epsilon_1\epsilon_2\big)
\nonumber
\\
&\ \ \
+t^7\frac{7\mu}{4\epsilon_2}
\big(40\mu^3+88\mu^2(\epsilon_1+\epsilon_2)+64\mu(\epsilon_1^2+\epsilon_2^2)
+108\mu\epsilon_1\epsilon_2+15(\epsilon_1^3+\epsilon_2^3)
\nonumber\\
&\quad\quad\quad\quad\quad
+32\epsilon_1\epsilon_2(\epsilon_1+\epsilon_2)\big)+\mathcal{O}(t^9),
\nonumber
\end{align}
\begin{align}
S_{1,2}(t)&=t^2\frac{\mu(2\mu+\epsilon_2)}{\epsilon_2}
+t^4\frac{\mu(2\mu+\epsilon_2)(4\mu+2\epsilon_1+3\epsilon_2)}{\epsilon_2}
\nonumber\\
&\ \ \
+t^6\frac{3\mu(2\mu+\epsilon_2)\big(20\mu^2+2\mu(12\epsilon_1+17\epsilon_2)
+7\epsilon_1^2+15\epsilon_2^2+17\epsilon_1\epsilon_2\big)}{4\epsilon_2}
+\mathcal{O}(t^8).
\nonumber
\end{align}

\noindent{\textbf{\underline{The 2-backbone phase function $S_2(x,t)$}}}\\
The solutions $S_{2,p}(t)$ ($p=0,1,2$) of equation (\ref{hier_b_gen}) take form
\begin{align}
S_{2,0}(t)&=-t^2\frac{2\mu}{\epsilon_1\epsilon_2}
-t^4\frac{8\mu(2\mu+\epsilon_1+\epsilon_2)}{\epsilon_1\epsilon_2}
\nonumber\\
&\ \ \
-t^6\frac{3\mu}{2\epsilon_1\epsilon_2}
\big(64\mu^2+78\mu(\epsilon_1+\epsilon_2)
+23(\epsilon_1^2+\epsilon_2)^2+39\epsilon_1\epsilon_2\big)
\nonumber\\
&\ \ \
-t^8\frac{\mu}{\epsilon_1\epsilon_2}
\big(512\mu^3+1084\mu^2(\epsilon_1+\epsilon_2)
+762\mu(\epsilon_1^2+\epsilon_2^2)+1304\mu\epsilon_1\epsilon_2
\nonumber\\
&\quad\quad\quad\quad\quad
+174(\epsilon_1^3+\epsilon_2^3)+381\epsilon_1\epsilon_2(\epsilon_1+\epsilon_2)
\big)+\mathcal{O}(t^{10}),
\nonumber
\end{align}
\begin{align}
S_{2,1}(t)&=
t^3\frac{4\mu}{\epsilon_2}+t^5\frac{24\mu(2\mu+\epsilon_1+\epsilon_2)}
{\epsilon_2}
\nonumber\\
&\ \ \
+t^7\frac{6\mu}{\epsilon_2}
\big(64\mu^2+78\mu(\epsilon_1+\epsilon_2)
+23(\epsilon_1^2+\epsilon_2^2)+39\epsilon_1\epsilon_2\big)
+\mathcal{O}(t^9),
\nonumber
\end{align}
\begin{align}
S_{2,2}(t)&=
t^2\frac{\mu}{\epsilon_2}
+t^4\frac{\mu(16\mu+3\epsilon_1+8\epsilon_2)}{\epsilon_2}
\nonumber\\
&\ \ \
+t^6\frac{3\mu}{4\epsilon_2}
\big(192\mu^2+2\mu(61\epsilon_1+117\epsilon_2)
+13\epsilon_1^2+69\epsilon_2^2+61\epsilon_1\epsilon_2\big)
+\mathcal{O}(t^8).
\nonumber
\end{align}

\noindent{\textbf{\underline{The 3-backbone phase function $S_3(x,t)$}}}\\
The solutions $S_{3,p}(t)$ ($p=0,1,2$) of equation (\ref{hier_b_gen}) take form
\begin{align}
S_{3,0}(t)&=
-t^4\frac{4\mu}{\epsilon_1\epsilon_2}
-t^6\frac{116\mu(2\mu+\epsilon_1+\epsilon_2)}{3\epsilon_1\epsilon_2}
\nonumber\\
&\ \ \
-t^8\frac{6\mu}{\epsilon_1\epsilon_2}
\big(148\mu^2+178\mu(\epsilon_1+\epsilon_2)
+52(\epsilon_1^2+\epsilon_2^2)+89\epsilon_1\epsilon_2\big)
+\mathcal{O}(t^{10}),
\nonumber
\end{align}
\begin{align}
S_{3,1}(t)&=
t^5\frac{14\mu}{\epsilon_2}
+t^7\frac{174\mu(2\mu+\epsilon_1+\epsilon_2)}{\epsilon_2}
+\mathcal{O}(t^9),
\nonumber
\end{align}
\begin{align}
S_{3,2}(t)&=
t^4\frac{4\mu}{\epsilon_2}
+t^6\frac{2\mu(58\mu+15\epsilon_1+29\epsilon_2)}{\epsilon_2}
+\mathcal{O}(t^8).
\nonumber
\end{align}


\end{document}